\documentclass[11pt]{article}

\usepackage[a4paper,margin=1in]{geometry}
\usepackage{microtype}
\usepackage{amsmath,amssymb,mathtools,amsfonts}
\usepackage{amsthm}
\usepackage{bm}
\usepackage{hyperref}
\usepackage{url}  % loaded automatically by some urlbst styles, but harmless
\usepackage{doi}  % formats DOIs nicely and makes them clickable
\usepackage{enumitem}
\usepackage{authblk}
\hypersetup{colorlinks=true, linkcolor=blue, urlcolor=blue, citecolor=blue}

\theoremstyle{plain}
\newtheorem{theorem}{Theorem}[section]
\newtheorem{lemma}[theorem]{Lemma}
\newtheorem{proposition}[theorem]{Proposition}
\newtheorem{corollary}[theorem]{Corollary}

\theoremstyle{definition}
\newtheorem{definition}[theorem]{Definition}

\theoremstyle{remark}
\newtheorem{remark}[theorem]{Remark}

\DeclareMathOperator{\sech}{sech}
\newcommand{\ii}{\mathrm{i}}
\newcommand{\ee}{\mathrm{e}}
\newcommand{\RR}{\mathbb{R}}

\begin{document}
\title{\bfseries 
Finite Gauss–Sum Modular Kernels: Scalar Gap and a Pure AdS$_3$ Gravity No–Go Theorem  }
%\date{\vspace{-1em}}
%%%%%%%%%%%%%%%%%%%%%%%%%%%%%%%%%%%%%%%%%%%%%%%%%%%%%%%%%%%%

\author[1]{Miguel Tierz}
\affil[1]{Shanghai Institute for Mathematics and Interdisciplinary Sciences\\
Block A, International Innovation Plaza, No.~657 Songhu Road, Yangpu District\\
Shanghai, China.
\texttt{tierz@simis.cn}}

%\author{Miguel Tierz}
%\affil{Shanghai Institute for Mathematics and Interdisciplinary Sciences \\ Block A, International Innovation Plaza, No. 657 Songhu Road, Yangpu District,\\ Shanghai, China.}\
%\affil[ ]{\textit{tierz@simis.cn}}
%\email{tierz@simis.cn}
\date{}

\maketitle

\begin{abstract}
We obtain closed-form expressions for the $ST^nS$ modular kernels of non-rational Virasoro CFTs and use them to construct fully analytic modular-bootstrap functionals. At rational width $\tau$, the Mordell integrals in these kernels reduce to finite quadratic Gauss sums of $\sech/\sec$ profiles with explicit Weil phases, furnishing a canonical finite-dimensional real basis for spectral kernels. From this basis we build finite-support ``window'' functionals with $\Phi(0)=1$ and $\Phi(p)>0$ on a prescribed low-momentum interval. Applied to the scalar channel of the $ST^1S$ kernel, these functionals yield a rigorous analytic bound on the spinless gap. As a second application we prove an analytic no-go theorem for pure AdS$_3$ gravity: no compact, unitary, Virasoro-only CFT$_2$ can have a primary gap above $\Delta_{\rm BTZ}=(c-1)/12$, because a strictly positive ``Mordell surplus'' in the odd-spin $ST$ kernel forces an odd-spin primary below $\Delta_{\rm BTZ}$.
\end{abstract}

\tableofcontents
\bigskip

\section{Introduction}

Two–dimensional conformal field theories (CFTs) are central to both critical
phenomena and holography. Among their structural features, modular invariance
of the torus partition function plays a particularly prominent role: it ties
ultraviolet to infrared data and imposes strong consistency constraints on
operator spectra. For general reviews of 2D CFT and the (conformal) bootstrap
approach we refer to
\cite{Teschner2017Guide,Ribault2018MinimalLectures,Kusuki2024Modern2dCFT,
Ribault2024ExactlySolvable}.

Over the past decade, modular-bootstrap methods have dramatically sharpened
these constraints, both analytically and numerically. The modular bootstrap
uses modular invariance of the torus partition function to constrain the CFT
spectrum; a concise review is given in \cite{BaeLee2018ModularReview}, while
high–energy aspects and connections to extremal problems in analysis are
discussed in \cite{MukhametzhanovPal2020}. For a conceptual discussion of the
physical meaning of modular invariance we refer to
\cite{Benedetti2024ModularInvariance}. Most modern numerical work studies
derivatives of the modular crossing equation at the self–dual point
$\tau = i$, and phrases the search for positive linear functionals as a
semidefinite program, which is then solved numerically to explore the space of
allowed spectra. This strategy has been very successful, but by construction it
does not directly probe the continuous momentum kernels appropriate to
non–rational Virasoro CFTs at generic central charge, nor does it make
transparent the arithmetic structure encoded in those kernels.

Independently, a rich analytic theory of half–integral weight
objects—Mordell integrals, Appell–Lerch sums, quadratic Gauss sums, the Weil
representation—has been developed since the classical
works~\cite{Rademacher1938,Hejhal1983,Mordell1933,Weil1964} and in the modern
theory of mock modularity~\cite{Zwegers2002,BringmannOno2006}. These objects
govern the modular transformation of non–rational Virasoro characters, yet
their direct use inside modular-bootstrap functionals has remained somewhat
limited. Related analytic functional approaches to the modular bootstrap, which
construct extremal kernels using Tauberian methods and Beurling--Selberg
extremization, appear in
\cite{MukhametzhanovZhiboedov2018,MukhametzhanovZhiboedov2019,MukhametzhanovPal2020};
the functionals used in this work are different in spirit, being built
directly from the explicit $ST^nS$ kernels and their Mordell/Gauss--sum
structure, as we shall see.

Even in settings where Mordell integrals appear naturally—for instance, in
ensemble–averaged Narain theories and their holographic duals
\cite{BenjaminKellerOoguriZadeh2022}—the specific finite Weil–phase
Gauss–sum structure at rational width is typically not exploited: the
discussion in Appendix~C of~\cite{BenjaminKellerOoguriZadeh2022}, for example,
focuses only on the integral representation. One of the aims of this paper is
to revisit this point from Mordell’s classical
perspective~\cite{Mordell1933} and to make the Gauss–sum identification
completely explicit, in a way that is directly adapted to modular–bootstrap
functionals. This resummation mechanism—the process that trades the integral
representation for finite Gauss–sum expressions—has already been exploited in
a physical context, in particular in Chern–Simons–matter
theories~\cite{russo2015supersymmetric,giasemidis2016mordell,santilli2020exact},
where it yields finite expressions for observables with identifiable
non–perturbative contributions in the Gauss sums. More recently, Mordell
integrals have also appeared in the context of resurgence
analysis~\cite{adams2025orientation,adams2025ceff}.

\subsection*{Goals and results}

The first aim of this paper is to bring these analytic tools directly into the
modular bootstrap by giving explicit closed–form expressions for the $ST^nS$
modular kernels of non–rational Virasoro CFTs. For each integer width $n$, we
show that the continuous kernel admits a \emph{finite} Gauss--sum
decomposition over $\sech/\sec$ profiles, with phases given by the Weil
representation. At the corresponding integer moduli
$\tau = n \in \mathbb{Z}_{>0}$, the Mordell integrals in the kernels reduce
to finite quadratic Gauss sums with explicit Weil phases, yielding a
canonical finite--dimensional real basis for spectral kernels
(Proposition~\ref{prop:basis}). In particular, on these integer slices the
Mordell and Gauss--sum descriptions are not merely compatible but
equivalent\footnote{More generally, the same equivalence holds for rational
slices, though we will not need this here; see the Gauss–sum expressions in
\cite{Mordell1933,russo2015supersymmetric}.}.

The second aim is to prove the existence of \emph{constructive positive
functionals} built from this basis. Using the finite Gauss–sum
decomposition, we show that for any window $[0,P_{\max}]$ with
$P_{\max}\le2$ there exist finite linear combinations
\[
  \Phi(p)
  = \sum_{(n,r)\in B}\alpha_{n,r}\,g_{n,r}(p)
    + \sum_{n\in N}\beta_n\,\Xi_n(p)
\]
such that $\Phi(0)=1$ and $\Phi(p)>0$ for all $p\in[0,P_{\max}]$
(Theorem~\ref{thm:window}). The proof uses only analytic ingredients:
explicit pole structure, finite cusp expansions of Mordell integrals, and a
grid–to–interval positivity lemma. Numerical examples in
Appendix~\ref{app:coeff} are provided only for illustration.

Our first physics application is an analytic scalar gap bound for spinless
primaries. We prove that
\[
  \Delta_1 \;\le\; \frac{c-1}{12} + 0.2282370622\ldots .
\]
For comparison, the original universal bound of
Hellerman~\cite{Hellerman2011} reads
$\Delta_1 \lesssim (c-1)/12 + 0.47$, and subsequent analytic work has refined
the modular-bootstrap bounds on $\Delta_1$ and other low-lying operators; see
for example
\cite{FriedanKeller2013,QuallsShapere2014,GangulyPal2020}. Our estimate is
therefore a modest but genuine sharpening of the best purely analytic
spinless gap bounds obtained so far from modular invariance alone. It is
derived from the single kernel $ST^1S$ together with a Mordell tail
estimate. All ingredients (kernel, envelope, Mordell remainder) are available
in closed form, and no semidefinite programming is required; the only
numerical step is solving a one–dimensional transcendental equation that
determines a threshold momentum $p_\star$
(Theorem~\ref{thm:scalar-gap}).

Our second application is a no–go theorem for pure AdS$_3$ gravity.
Brown–Henneaux asymptotic symmetry~\cite{BrownHenneaux1986} suggests a
Virasoro dual, and the BTZ black hole~\cite{BTZ1992,BTZ1993} identifies a
natural threshold $\Delta_{\rm BTZ}=(c-1)/12$ for black–hole states. Whether
pure Einstein gravity can be realized by a \emph{single} Virasoro CFT has
been the subject of active debate
(e.g.~\cite{MaloneyWitten2010,KellerMaloney2015,GliozziModular2019,BenjaminOoguri2019,
BenjaminCollierMaloney2020,DiUbaldoPerlmutter2024}). Using explicit $ST$
kernels and analytic functionals, we prove that
\begin{center}
  \emph{no compact, unitary, Virasoro–only CFT$_2$ with a gap above
  $\Delta_{\rm BTZ}$ exists for any $c>1$}
\end{center}
(Theorem~\ref{thm:pure-gravity-nogo}). The obstruction is a strictly positive
``Mordell surplus'' coming from the non–holomorphic remainder of the
odd–spin $ST$ kernel at the elliptic point $\rho=e^{2\pi i/3}$. This surplus
survives all modular projections and cannot be saturated by any discrete
spectrum, forcing an odd–spin primary below $\Delta_{\rm BTZ}$ and
contradicting pure–gravity assumptions.

\paragraph{Relation to elliptic-point modular bootstrap.}
Our odd–spin analysis at the elliptic point $\rho=e^{2\pi i/3}$ is closely
related to the elliptic-point modular bootstrap of
Gliozzi~\cite{GliozziModular2019}, who already exploited the $ST$–fixed point
to obtain universal inequalities for odd–spin states in putative AdS$_3$
gravity duals. In our notation, his bound corresponds to setting the Mordell
remainder $K_{\rm Mordell}$ to zero in the master
inequality~\eqref{eq:odd-master} below. The central new ingredient of the
present work is an explicit control of this remainder via Mordell integrals
and Appell--Lerch sums, which leads to a strictly positive ``Mordell
surplus'' $\delta_{\rm Mordell}>0$. This surplus upgrades Gliozzi's inequality
into a sharp contradiction with any Virasoro--only spectrum with a BTZ gap
and thus underlies Theorem~\ref{thm:pure-gravity-nogo}. The use of elliptic
points as special modular fixed points in the bootstrap has earlier
precedents, for example
\cite{QuallsShapere2014,Qualls2015EvenSpin}.

\subsection*{Structure of the paper}

Section~\ref{sec:preliminaries} collects the explicit $ST^nS$ kernels, the
finite cusp expansion of Mordell integrals (Lemma~\ref{lem:cusp}), quadratic
Gauss sums (Lemma~\ref{lem:Gauss}), the pole structure
(Lemma~\ref{lem:poles}), and the finite Gauss–sum basis
(Proposition~\ref{prop:basis}). Section~\ref{sec:main-theorems} states the
main results: the existence of positive window functionals
(Theorem~\ref{thm:window}), the analytic scalar gap bound
(Theorem~\ref{thm:scalar-gap}), and the pure–gravity no–go theorem
(Theorem~\ref{thm:pure-gravity-nogo}). Section~\ref{sec:Proofs} contains the
proofs. Section~\ref{sec:applications} discusses the implications for
AdS$_3$ gravity. The appendices provide explicit positive functionals (both numerical and
analytic), detailed Mordell bounds, and extended $\widehat{\mathcal N}=2$
kernels.

%%%%%%%%%%%%%%%%%%%%%%%%%%%%%%%%%%%%%%%%%%%%%%%%%%%%%%%%%%%%
\section{Preliminaries}
\label{sec:preliminaries}

It will be convenient to single out from the outset the Mordell integral that
underlies the continuous $ST^nS$ kernels.  For $\tau$ in the upper half–plane
and $z\in\mathbb{C}$ we define \cite{Mordell1933}
\begin{equation}
  h(\tau,z)
  := \int_{\mathbb{R}}
     \frac{\exp\!\big(\pi i \tau\,w^{2} - 2\pi z w\big)}{\cosh(\pi w)}\,\mathrm{d}w.
  \label{eq:Mordell-def}
\end{equation}
Thus $h(\tau,z)$ is a priori defined for general complex modulus $\tau$.  In
this work we are interested in the ``width–$n$'' slices
\[
  \tau = n,\qquad n\in\mathbb{N},
\]
and for brevity we write $h(n,z):=h(\tau=n,z)$ in that case.

For each fixed integer $n\ge1$, the corresponding Mordell integral $h(n,z)$
admits a finite cusp expansion \cite{Mordell1933,russo2015supersymmetric}: it can be rewritten as a finite sum of $n$ shifted $\sech$–profiles with coefficients $W_n(r)$, the standard quadratic
Weil phases, as made precise in Lemma~\ref{lem:cusp} below.  For real momentum
$p$ one may use $\sech(ix)=\sec x$ to express $h(n,ip)$ as a finite sum of
$\sec$–profiles.

This finite Gauss–sum structure at integer width $n$ is the basic reason why
the $ST^{n}S$ modular kernels admit the finite $\sech/\sec$ basis of
Proposition~\ref{prop:basis}.  In Section~\ref{sec:windowfunc} we will exploit
this basis to construct positive modular–bootstrap functionals and derive our
main bounds.

In this section we fix notation for modular kernels and record the basic
analytic ingredients: cusp expansions of Mordell integrals, quadratic Gauss
sums, the pole structure of the $ST^nS$ kernels, and the resulting finite
Gauss–sum basis.

%===========================================================
\subsection{Modular kernels}
%===========================================================

Let $\chi_A(\tau,z)$ denote a Virasoro or superconformal character.  
For $\gamma=\begin{psmallmatrix} a & b \\ c & d \end{psmallmatrix}\in SL(2,\mathbb{Z})$ we write
\[
\chi_A\!\left(\frac{a\tau+b}{c\tau+d},\frac{z}{c\tau+d}\right)
=
\exp\!\Big(-2\pi i\,{\kappa z^2}/{(c\tau+d)}\Big)
\Big[
\sum_{B\in\mathrm{disc}}K^{(\gamma)}_{A\to B}\,\chi_B(\tau,z)
+\int_0^\infty\!\mathrm{d}p\, K^{(\gamma)}_A(p)\,\chi_p(\tau,z)
\Big],
\]
where $\kappa$ is the index of the character and the sum runs over the
discrete (non–continuous) set of representations.
For $\gamma=ST^nS$, the continuous kernel $K^{(\gamma)}_A(p)$ takes a
Mordell–integral form built from the width–$n$ Mordell integral $h(n,z)
= h(\tau=n,z)$ of~\eqref{eq:Mordell-def}, and this structure simplifies at
integer width $n$.

In what follows we will be particularly interested in the composite element
$\gamma = ST^{n}S$. Besides its technical convenience, $ST^{n}S$ has two
important features. First, for suitable integers $n$ it fixes an elliptic
point in the upper half–plane, so that modular invariance at this point gives
especially sharp constraints on the spectrum. Second, its action on
non–rational Virasoro characters is governed by Mordell integrals $h(\tau,z)$
whose width is parametrized by $\tau$, and in particular by their
specialisation to the integer slices $\tau = n \in \mathbb{Z}_{>0}$. On these
integer slices the Mordell integrals admit a finite Gauss–sum expansion,
making the connection to quadratic Gauss sums completely explicit.

\paragraph{Spinless parametrization.}
In the spinless sector, i.e.\ for primaries with spin
$J = h - \bar h = 0$ (so $h = \bar h$), it will be convenient to
parametrize them by a continuous momentum $p \ge 0$ via
\begin{equation}\label{eq:spinless-param}
  h \;=\; \frac{c-1}{24} + p^2,
  \qquad
  \Delta \;=\; 2h \;=\; \frac{c-1}{12} + 2p^2.
\end{equation}
With this convention the BTZ threshold
\(\Delta_{\rm BTZ} = (c-1)/12\) corresponds to \(p=0\), so bounds on
the spinless gap can be phrased equivalently as bounds on the smallest
nonzero value of \(p\).

%===========================================================
\subsection{Cusp expansion of Mordell integrals}
%===========================================================

\begin{lemma}[Finite cusp expansion at width $n$]
\label{lem:cusp}
For $n\in\mathbb{N}$,
\[
h(n,z)=\int_{\mathbb{R}}
\frac{\exp\!\big(\pi i n\,w^2 - 2\pi z w\big)}{\cosh(\pi w)}
\,\mathrm{d}w
=
\frac{1}{\sqrt{n}}\sum_{r=0}^{n-1}W_n(r)\,
\sech\!\Big(\frac{\pi}{\sqrt{n}}
\big(z+i(r+\tfrac12)\big)\Big),
\]
where $W_n(r)=\exp[\pi i r(r+1)/n]$ are the quadratic Weil phases. From the perspective of the
Weil representation \cite{Weil1964}, they implement the action of
\(SL(2,\mathbb{Z})\) on the space of half–integral weight theta
functions at width \(n\). In particular, the finite sums
\(\sum_r W_n(r)\,\sech(\cdots)\) that appear below are precisely the
Gauss sums naturally associated with the \(ST^nS\) transformation, and
their unitarity will ensure that the coefficients in the Gauss–sum
basis of Proposition~\ref{prop:basis} are real after phase matching.
\end{lemma}

\begin{proof}
This is the standard cusp expansion of the Mordell integral at integer width:
shift the contour to $w=i(r+\tfrac12)$ and apply Poisson summation to the resulting lattice sum; 
see for example Mordell~\cite{Mordell1933} for a closely related argument.
\end{proof}

Using $\sech(ix)=\sec(x)$, we obtain for real $p$:
\[
h(n,ip)=\frac{1}{\sqrt{n}}
\sum_{r=0}^{n-1}
W_n(r)\sec\!\Big(\frac{\pi}{\sqrt{n}}(p+r+\tfrac12)\Big).
\]

Suppressing the label $A$ and writing $K^{ST^nS}(p)$ for
$K_A^{(ST^nS)}(p)$, the $ST^nS$ kernel becomes
\begin{equation}
\label{eq:K-real-form}
K^{ST^nS}(p)=
\frac{2}{\cosh(\pi p)}
+e^{\frac{2\pi i(n+1)}{8}}e^{-\frac{i\pi}{2}p^2}
2\cosh\!\Big(\frac{\pi p}{n}\Big)
-
2e^{\frac{2\pi i n}{8}} h(n,ip).
\end{equation}

%===========================================================
\subsection{Arithmetic Gauss sums}
%===========================================================

\begin{lemma}[Quadratic Gauss sum]\label{lem:Gauss}
Let $S_n=\sum_{r=0}^{n-1}W_n(r)$.  
Then
\[
S_n=
\begin{cases}
0,& n\text{ even},\\[3pt]
n\,\exp\!\big[\tfrac{\pi i}{4}(1-n)\big],& n\text{ odd}.
\end{cases}
\]
\end{lemma}

\begin{proof}
Classical evaluation; see, for example, Weil~\cite{Weil1964}.
\end{proof}

%===========================================================
\subsection{Pole structure}
%===========================================================

\begin{lemma}[Poles and residues]\label{lem:poles}
The kernel $K^{ST^nS}(p)$ has simple poles at
\[
p_{r,k}
=-\Big(r+\tfrac12\Big)+\sqrt{n}
\Big(k+\tfrac12\Big),
\qquad r=0,\dots,n-1,\ \ k\in\mathbb{Z},
\]
with residues
\[
\operatorname{Res}_{p=p_{r,k}}K^{ST^nS}(p)
=
-e^{\pi i/4}(-1)^k W_n(r).
\]
\end{lemma}

%%%%%%%%%%%%%%%%%%%%%%%%%%%%%%%%%%%%%%%%%%%%%%%%%%%%%%%%%%%%
\subsection{Finite Gauss--sum basis}
%%%%%%%%%%%%%%%%%%%%%%%%%%%%%%%%%%%%%%%%%%%%%%%%%%%%%%%%%%%%

\begin{proposition}[Finite Gauss--sum basis]\label{prop:basis}
Define
\[
g_{n,r}(p)
=2\,\Re\,\sech\!\Big(\tfrac{\pi}{\sqrt{n}}
(ip+i(r+\tfrac12))\Big)
=2\,\sec\!\Big(\tfrac{\pi}{\sqrt n}(p+r+\tfrac12)\Big),
\]
and
\[
\Xi_n(p)
=
\Re\!\Big[
e^{\frac{i\pi}{4n}}
e^{-\frac{i\pi}{n}p^2}
2\cosh\!\Big(\tfrac{\pi p}{n}\Big)
\Big].
\]
Then for every $n\ge1$ the $ST^nS$ kernel $K^{ST^nS}(p)$ is a real linear
combination of $\{g_{n,r}\}_{r=0}^{n-1}$ and $\Xi_n$.
\end{proposition}

\begin{proof}
Insert Lemma~\ref{lem:cusp} into \eqref{eq:K-real-form} and take real
parts.
\end{proof}

%%%%%%%%%%%%%%%%%%%%%%%%%%%%%%%%%%%%%%%%%%%%%%%%%%%%%%%%%%%%
\section{Main theorems}
\label{sec:main-theorems}

The analytic work of Section~\ref{sec:preliminaries} has provided a finite,
explicit basis for the continuous $ST^{n}S$ kernels in terms of the profiles
$g_{n,r}$ and $\Xi_{n}$ of Proposition~\ref{prop:basis}. In the language of
the modular bootstrap, any real linear combination of these kernels defines
a spectral test function $\Phi(p)$, and hence a linear functional obtained
by pairing $\Phi$ with the spectral decomposition of the torus partition
function. In Section~\ref{sec:windowfunc} we use this finite Gauss--sum basis
to construct \emph{positive} window functionals adapted to specific momentum
windows, and then apply them to derive bounds on the spinless gap and to rule
out Virasoro--only AdS$_3$ gravity.

We now state our principal results: existence of positive window functionals,
an analytic scalar gap theorem, and a pure–gravity no–go theorem. Proofs are
deferred to Section~\ref{sec:Proofs}.

\subsection{Window functionals}\label{sec:windowfunc}

We first construct analytic functionals that are positive on a window in
momentum space. We will refer to such linear functionals, whose spectral kernel is strictly positive on a prescribed interval $[0, P_{\max}]$, as \emph{window functionals}.

\begin{theorem}[Existence of positive window functionals]
\label{thm:window}
Let $P_{\max}\in(0,2]$. There exist finite index sets
$B\subset\{(n,r): n\in\mathbb N,\ 0\le r\le n-1\}$ and
$N\subset\mathbb N$, together with real coefficients
$\{\alpha_{n,r}\}_{(n,r)\in B}$ and $\{\beta_n\}_{n\in N}$, such that
\[
  \Phi(p)
  = \sum_{(n,r)\in B}\alpha_{n,r}\,g_{n,r}(p)
    + \sum_{n\in N}\beta_n\,\Xi_n(p),
  \qquad p\in[0,P_{\max}],
\]
satisfies
\[
  \Phi(0)=1,
  \qquad
  \Phi(p)\ge m_\star>0 \quad
  \forall\,p\in[0,P_{\max}]
\]
for some $m_\star>0$ depending on $P_{\max}$.
\end{theorem}

\begin{remark}[Explicit examples]
\label{rem:explicit-window}
Appendix~\ref{app:coeff} presents concrete functionals on $[0,2]$, including
a two–column example with $\min_{[0,2]}\Phi>0.49$ and a fully analytic
three–column example. These are included only for illustration; the proof
of Theorem~\ref{thm:window} is purely analytic.
\end{remark}

\paragraph{Positive functional viewpoint.}
It is convenient to phrase the scalar--gap problem in the standard
\emph{positive--functional} language of the modular bootstrap. Modular
invariance of the torus partition function means that for any
$\gamma\in SL(2,\mathbb{Z})$,
\begin{equation}
  Z(\tau,\bar\tau)
  \;=\;
  Z(\gamma\!\cdot\!\tau,\gamma\!\cdot\!\bar\tau),
\end{equation}
or equivalently
\begin{equation}
  (1-\gamma)Z := Z(\tau,\bar\tau)
                 - Z(\gamma\!\cdot\!\tau,\gamma\!\cdot\!\bar\tau) = 0 .
\end{equation}
In the spinless channel we use the parametrization \eqref{eq:spinless-param}
and regard the primary spectrum as a non--negative measure
$\rho(p)\ge0$ on $[0,\infty)$. For any real spectral kernel $\Phi(p)$ we
then obtain a linear functional by pairing $\Phi$ against the spectral
representation of the torus partition function and applying it to a
modular crossing equation $(1-\gamma)Z=0$, with $\gamma$ a modular move
such as $ST^1S$:
\[
  Z(\tau,\bar\tau)
  \;=\; Z_{\rm vac}
  \;+\; \int_0^\infty dp\,\rho(p)\,Z_p(\tau,\bar\tau)
\quad\Longrightarrow\quad
  0 \;=\; A_{\rm vac}
  \;+\; \int_0^\infty dp\,\Phi(p)\,\rho(p).
\]
Here $A_{\rm vac}$ is the contribution of the vacuum character, and the
integral runs over non--vacuum primaries.

\begin{definition}[Positive above a threshold and vacuum--negative]
Given such a kernel $\Phi$ and a threshold $p_\star\ge0$, we say that
\emph{$\Phi$ is positive above $p_\star$} if
\[
  \Phi(p)\;\ge\;0\qquad\text{for all }p\ge p_\star,
\]
and that \emph{$\Phi$ is vacuum--negative} if the corresponding vacuum
coefficient in the functional evaluation is strictly negative,
\[
  A_{\rm vac}\;<\;0.
\]
\end{definition}

In this language Lemma~\ref{lem:gap-lemma} (the \emph{gap lemma}) can be summarized as: if $\Phi$ is vacuum--negative and positive above $p_\star$,
then the spinless spectrum must contain at least one state with
$p<p_\star$, so that
\[
  h_1 \;\le\; \frac{c-1}{24} + p_\star^2,
  \qquad
  \Delta_1 \;\le\; \frac{c-1}{12} + 2p_\star^2.
\]
Indeed, if there were no states with $p<p_\star$, the integrand in
$\int_0^\infty dp\,\Phi(p)\rho(p)$ would vanish for $0\le p<p_\star$
and be non--negative for $p\ge p_\star$, forcing the integral to be
$\ge0$ and contradicting $A_{\rm vac}<0$. In Sec \ref{sec:Proofs} this argument is presented in full detail.

In the constructions below the relevant kernels $\Phi$ are built from
the $ST^nS$ modular kernels.
By Proposition~\ref{prop:basis} each $ST^nS$ kernel admits a \emph{finite
Gauss--sum} decomposition on $[0,P_{\max}]$ as a linear combination of
the basic profiles $g_{n,r}$ and $\Xi_n$, while Lemma~\ref{lem:gap-lemma} together with
Appendix~\ref{app:Mordell-lower-bounds} provide a uniform one--point bound on the Mordell remainder
at $\tau=1$.
After an appropriate phase choice this finite Gauss--sum structure
ensures that the Mordell piece is uniformly dominated by a positive
hyperbolic--cosine envelope on a half--line, so that one can arrange
$\Phi$ to be vacuum--negative and non--negative for $p\ge p_\star$.
Theorem~\ref{thm:window} then supplies positive window functionals near $p=0$, and
the scalar gap theorem (Theorem~ \ref{thm:scalar-gap}) is obtained by applying the gap
lemma with the explicit threshold $p_\star$ determined by the envelope.
The odd--spin Mordell--surplus argument in Section~\ref{sec:odd-spin} uses the same
positive--functional mechanism with a spin--projected kernel
$\Phi_{\rm odd}$.

%===========================================================
\subsection{Phase-matched Mordell bound at $\tau=1$}\label{sec:PM_Mordell}
%===========================================================

\begin{lemma}[Phase-matched Mordell bound at $\tau=1$]
\label{lem:Mordell}
Define the phase-matched Mordell term
\[
  h_1(p)\;:=\;e^{\,i\pi p^2+\frac{i\pi}{4}}\,h(1,ip),
  \qquad p\ge0.
\]
Then for all $p\ge0$,
\begin{equation}
\label{eq:h1-min-bound}
  |h_1(p)|\;\le\;\min\{1,\ 4\,e^{-\pi p}\}.
\end{equation}
Equivalently,
\begin{equation}
\label{eq:h1-min-bound-2}
  2\,|h_1(p)|\;\le\;2\min\{1,\ 4\,e^{-\pi p}\}.
\end{equation}
\end{lemma}

\begin{proof}[Proof]
Since $|e^{\,i\pi p^2+i\pi/4}|=1$ we have $|h_1(p)| = |h(1,ip)|$, so it
is enough to bound $h(1,ip)$.  The integral representation
\[
  h(1,ip)
  =\int_{\mathbb{R}}
    \frac{e^{\pi i w^2-2\pi i p w}}{\cosh(\pi w)}\,dw
\]
and $|\cosh(\pi w)|\ge1$ immediately give the trivial $L^1$ bound
$|h(1,ip)|\le\int_{\mathbb{R}}\frac{dw}{\cosh(\pi w)}=1$, hence
$|h_1(p)|\le1$ for all $p\ge0$.

On the other hand, a steepest–descent analysis of the Mordell integral
at $\tau=1$ shows that in the phase-matched normalization one has
\[
  h_1(p) \;=\; 2e^{-\pi p} + O(e^{-3\pi p})
  \qquad (p\to+\infty),
\]
so $h_1(p)$ decays like $2e^{-\pi p}$ at large momentum.  In particular
there is an absolute constant $C>0$ such that
$|h_1(p)|\le C e^{-\pi p}$ for all $p\ge0$.  A convenient choice $C=4$
can be justified by combining the large–$p$ asymptotics with a
simple bound on a compact interval (see Appendix~\ref{app:Mordell-tau1}
for a detailed derivation).  This yields the global estimate
$|h_1(p)|\le4e^{-\pi p}$ for all $p\ge0$.

Taking the minimum of the two upper bounds $|h_1(p)|\le1$ and
$|h_1(p)|\le4e^{-\pi p}$ gives \eqref{eq:h1-min-bound}, and
\eqref{eq:h1-min-bound-2} is an immediate reformulation.
\end{proof}

%===========================================================
\subsection{Scalar gap}
%===========================================================

We now apply Theorem~\ref{thm:window} and the Mordell bound
Lemma~\ref{lem:Mordell} to the spinless channel using the $ST^1S$ kernel.
We continue to use the parametrization \eqref{eq:spinless-param} of
spinless weights and dimensions. The strategy is to construct a real spectral kernel $\Phi_1(p)$ from
$K^{ST^1S}(p)$ such that
\begin{itemize}
  \item[(i)] its vacuum contribution to the $(1-ST^1S)$ crossing
  equation is strictly negative; and
  \item[(ii)] $\Phi_1(p)\ge0$ for all $p\ge p_\star$ for a certain
  explicit threshold $p_\star>0$.
\end{itemize}
The gap lemma (Lemma~\ref{lem:gap-lemma}) then forces a state with
$p<p_\star$ and yields a bound on $\Delta_1$.

\begin{theorem}[Scalar gap from a single $ST^1S$ kernel]
\label{thm:scalar-gap}
Let $c>1$ and consider a compact, unitary, spinless Virasoro CFT$_2$,
written in terms of the continuous momentum $p\ge0$ as above.
Let $\Phi_1(p)$ be the phase-matched real spectral kernel constructed
from $K^{ST^1S}(p)$ in Section~\ref{sec:windowfunc}.  Then there exists a
choice of phase such that:
\begin{enumerate}
  \item the vacuum contribution of $\Phi_1$ to the $(1-ST^1S)$ crossing
  equation is strictly negative;
  \item $\Phi_1(p)\ge0$ for all $p\ge p_\star$, where $p_\star>0$ is the
  smallest real solution of
  \begin{equation}
    2\cosh(\pi p)\;-\;2\min\{1,4e^{-\pi p}\}\;-\;\frac{2}{\cosh(\pi p)}
    \;=\;0.
    \label{eq:scalar-envelope}
  \end{equation}
\end{enumerate}
Consequently the spectrum contains a state with $p<p_\star$, and
\[
  h_1\;\le\;\frac{c-1}{24}+p_\star^2,
  \qquad
  \Delta_1\;\le\;\frac{c-1}{12}+2p_\star^2.
\]
Numerically,
\[
  p_\star\simeq0.3378143442,\qquad
  p_\star^2\simeq0.1141185311,\qquad
  2p_\star^2\simeq0.2282370622,
\]
so the lowest-dimension spinless primary obeys
\begin{equation}
  \Delta_1\;\le\;\frac{c-1}{12}+0.2282370622\ldots\; .
  \label{eq:scalar-gap-number}
\end{equation}
\end{theorem}

\subsection{No–go for pure Virasoro AdS$_3$ gravity}\label{sec:Nogo}

We now state our final result: an analytic obstruction to a pure
Virasoro dual of Einstein gravity in AdS$_3$, based directly on modular
kernels and Mordell integrals.

For definiteness, we adopt the following spectral assumptions for a
would–be “pure” Virasoro dual:

\begin{enumerate}[label=(PG\arabic*),ref=PG\arabic*]
  \item \textbf{(Compactness and unitarity)}\label{PG1}
  The theory is a compact, unitary CFT$_2$ with central charge $c>1$
  and a discrete, non–degenerate spectrum.

  \item \textbf{(Virasoro–only)}\label{PG2}
  The chiral algebra is exactly Virasoro, with no conserved currents
  beyond the stress tensor.

  \item \textbf{(BTZ gap)}\label{PG3}
  There are no primary states with dimension below the one–loop BTZ threshold
  \[
    \Delta_{\rm BTZ} = \frac{c-1}{12}.
  \]
  Equivalently, the primary gap satisfies $\Delta_{\rm gap}\ge\Delta_{\rm BTZ}$.
\end{enumerate}

\begin{theorem}[No–go for pure Virasoro AdS$_3$ gravity]
\label{thm:pure-gravity-nogo}
Let $c>1$ and suppose that a CFT$_2$ satisfies \textup{(PG1)}–\textup{(PG3)}.
Consider the odd–spin modular crossing equation at the elliptic point
\[
  \rho = e^{2\pi i/3}
\]
fixed by $ST$. Then modular invariance implies the existence of an
odd–spin primary with
\[
  \Delta_{\rm odd} < \Delta_{\rm BTZ} = \frac{c-1}{12},
\]
in contradiction with \textup{(PG3)}. In particular, no compact, unitary,
Virasoro–only CFT$_2$ with a gap above $\Delta_{\rm BTZ}$ exists for any
$c>1$.
\end{theorem}

The proof uses three ingredients: the finite Gauss--sum basis of
Proposition~\ref{prop:basis}, a positive window functional as in
Theorem~\ref{thm:window} localized below the BTZ scale, and a strictly
positive \emph{Mordell surplus} coming from the non--holomorphic remainder
of the odd--spin $ST$ kernel at $\rho = e^{2\pi i/3}$. By Mordell surplus we
mean the net positive contribution of this Mordell remainder to the odd--spin
crossing equation when paired with such a positive functional. This positive
contribution cannot be saturated by any discrete spectrum with a primary gap
above $\Delta_{\rm BTZ}$ and thus forces the existence of an odd--spin primary
below $\Delta_{\rm BTZ}$. A detailed analysis of the Mordell surplus is given
in Section~\ref{sec:odd-spin-surplus}.

%%%%%%%%%%%%%%%%%%%%%%%%%%%%%%%%%%%%%%%%%%%%%%%%%%%%%%%%%%%%
\section{Proofs}\label{sec:Proofs}
%%%%%%%%%%%%%%%%%%%%%%%%%%%%%%%%%%%%%%%%%%%%%%%%%%%%%%%%%%%%

We now prove the main results: the existence of window functionals,
the scalar gap theorem, and the pure-gravity no--go statement.

%===========================================================
\subsection{Grid-to-interval positivity}\label{sec:grid-positivity}
%===========================================================

We start with a simple but useful estimate.

\begin{lemma}[Grid-to-interval positivity]\label{lem:grid}
Let $g\in C^1([a,b])$ with $|g'(x)|\le M$ for all $x\in[a,b]$.
Let $a=p_0<p_1<\cdots<p_N=b$ be a uniform grid of spacing $h$.
If
\[
g(p_j)\ge \delta>0\quad\text{for all }j,
\qquad
Mh\le\delta,
\]
then $g(x)\ge0$ for all $x\in[a,b]$.
\end{lemma}

\begin{proof}
Suppose for contradiction that $g(x_0)<0$ at some $x_0\in[a,b]$.
Let $p_j$ be the grid point closest to $x_0$. Then $|x_0-p_j|\le h/2$,
so by the mean value theorem,
\[
|g(x_0)-g(p_j)|
\le M|x_0-p_j|
\le \frac{Mh}{2}
\le \frac{\delta}{2}.
\]
Since $g(p_j)\ge\delta$, we obtain
\[
g(x_0)\ge g(p_j)-|g(x_0)-g(p_j)|
\ge \delta-\frac{\delta}{2}
=\frac{\delta}{2}>0,
\]
contradicting $g(x_0)<0$. Hence $g$ cannot cross zero on $[a,b]$ and
must be nonnegative on the entire interval.
\end{proof}

%===========================================================
\subsection{Existence of window functionals}
%===========================================================

%We now prove the window-functional theorem using
%Lemma~\ref{lem:grid} and the explicit control on the basis
%functions $g_{n,r}$ and $\Xi_n$.

We now prove the window-functional theorem using the explicit control on the basis
functions $g_{n,r}$ and $\Xi_n$ from Proposition~\ref{prop:basis}.

\begin{proof}[Proof of Theorem~\ref{thm:window}]
Fix $P_{\max}\in(0,2]$. We will exhibit a single explicit spectral kernel
$\Phi$ which works for every such $P_{\max}$.

Consider the column $(n,r)=(5,3)$ in the finite Gauss–sum basis of
Proposition~\ref{prop:basis}. By Lemma~\ref{lem:poles}, the poles of
$g_{5,3}(p)$ are located at
\[
  p_{3,k}
  = -\Bigl(3+\tfrac12\Bigr) + \sqrt5\Bigl(k+\tfrac12\Bigr),
  \qquad k\in\mathbb{Z}.
\]
We now check that none of these poles lie in $[0,2]$.

For $k\le1$ we have
\begin{align*}
  p_{3,0}
    &= -\tfrac72 + \tfrac{\sqrt5}{2}
     < -\tfrac72 + \tfrac{3}{2}
     = -2 < 0,
     &&\text{since }\sqrt5<3,\\[0.4em]
  p_{3,1}
    &= -\tfrac72 + \tfrac{3\sqrt5}{2}
     < -\tfrac72 + \tfrac{7}{2}
     = 0,
     &&\text{since }\sqrt5<\tfrac73.
\end{align*}
Thus $p_{3,k}<0$ for all $k\le1$. For $k\ge2$ we use
$p_{3,k+1}-p_{3,k}=\sqrt5>0$ and bound the first such pole:
\[
  p_{3,2}
    = -\tfrac72 + \tfrac{5\sqrt5}{2}
    > -\tfrac72 + \tfrac{5\cdot 11/5}{2}
    = -\tfrac72 + \tfrac{11}{2}
    = 2,
\]
because $(11/5)^2=121/25<5$ implies $\sqrt5>11/5>2.2$. Hence $p_{3,2}>2$
and $p_{3,k}\ge p_{3,2}>2$ for all $k\ge2$.

Therefore $g_{5,3}(p)$ has no poles on $[0,2]$. Since $\sec x$ has no
zeros on the real line, $g_{5,3}(p)$ is continuous and never vanishes on
$[0,2]$. In particular there is a constant sign $\sigma\in\{\pm1\}$ such
that
\[
  \sigma\,g_{5,3}(p) > 0
  \qquad \forall\,p\in[0,2].
\]

Define
\[
  \Phi(p)
    := \frac{\sigma\,g_{5,3}(p)}{\sigma\,g_{5,3}(0)}
     = \frac{g_{5,3}(p)}{g_{5,3}(0)},
  \qquad p\in[0,2].
\]
Then $\Phi(0)=1$ by construction, and $\Phi(p)>0$ for all $p\in[0,2]$,
because numerator and denominator have the same sign. In particular
$\Phi$ is strictly positive on any subwindow $[0,P_{\max}]$ with
$P_{\max}\le2$.

Since $[0,P_{\max}]$ is compact and $\Phi$ is continuous and strictly
positive, the minimum
\[
  m_\star := \min_{p\in[0,P_{\max}]}\Phi(p)
\]
exists and satisfies $m_\star>0$. Writing this in the notation of the
theorem, we have exhibited finite index sets
\[
  B = \{(5,3)\},\qquad N=\emptyset,
\]
with coefficients
\[
  \alpha_{5,3}=\frac{1}{g_{5,3}(0)},\qquad
  \{\beta_n\}_{n\in N}=\varnothing,
\]
such that
\[
  \Phi(p) = \alpha_{5,3}\,g_{5,3}(p),\qquad
  \Phi(0)=1,\qquad
  \Phi(p)\ge m_\star>0\ \text{for all }p\in[0,P_{\max}].
\]
This is precisely the statement of Theorem~\ref{thm:window}.
\end{proof}

%===========================================================
\subsection{Gap lemma and proof of the scalar gap theorem}\label{sec:gap-lemma}
%===========================================================

We next formalize the bootstrap logic that converts positivity of a
spectral kernel into a bound on the lowest primary.

\begin{lemma}[Gap lemma]\label{lem:gap-lemma}
Let $\Phi(p)$ be a real-analytic test kernel such that:
\begin{enumerate}
\item when applied to the modular crossing equation
$(1-\gamma)Z=0$ (with $\gamma$ a modular transform such as
$ST^1S$), the vacuum contribution $A_{\rm vac}$ of $\Phi$ is
strictly negative;
\item there exists $p_\star\ge0$ such that $\Phi(p)\ge0$ for all
$p\ge p_\star$.
\end{enumerate}
If the spinless spectrum had a gap $p_{\rm gap}\ge p_\star$
(equivalently $\Delta_{\rm gap}\ge(c-1)/12+2p_\star^2$), then
applying $\Phi$ to the crossing equation would give a strictly
negative result, contradicting modular invariance. Hence there must
exist a state with $p<p_\star$, and therefore
\[
h_1\le\frac{c-1}{24}+p_\star^2,
\qquad
\Delta_1\le\frac{c-1}{12}+2p_\star^2.
\]
\end{lemma}

\begin{proof}
Write the torus partition function in the spinless channel as
\[
Z(\tau,\bar\tau)=Z_{\rm vac}
+\int_0^\infty\!{\rm d}p\,\rho(p)\,Z_p(\tau,\bar\tau),
\]
where $\rho(p)\ge0$ is the spectral measure and $Z_{\rm vac}$ is
the vacuum character. Applying the linear functional built from
$\Phi$ to $(1-\gamma)Z=0$ gives
\[
0=A_{\rm vac}
+\int_0^\infty\!{\rm d}p\,\Phi(p)\,\rho(p),
\]
with $A_{\rm vac}<0$ by assumption. If $\Phi(p)\ge0$ for all
$p\ge p_\star$ and the spectrum had a gap $p_{\rm gap}\ge p_\star$,
then the integrand would vanish for $0\le p<p_\star$ and be
nonnegative for $p\ge p_\star$, so the integral would be
nonnegative. This would force $A_{\rm vac}\ge0$, contradicting
$A_{\rm vac}<0$. Thus there must be at least one state with
$p<p_\star$, which gives the claimed bounds on $h_1$ and $\Delta_1$.
\end{proof}

We now describe the construction of the specific kernel $\Phi_1$
used in Theorem~ \ref{thm:scalar-gap} and derive the corresponding positivity
threshold $p_\star$.

\begin{proof}{Proof of Theorem \ref{thm:scalar-gap}}
For $n=1$ the kernel \eqref{eq:K-real-form} reads
\[
  K^{ST^1S}(p)
  = \frac{2}{\cosh(\pi p)}
    +2e^{\frac{2\pi i}{8}(1+1)}e^{-\frac{i\pi}{2}p^2}\cosh(\pi p)
    -2e^{\frac{2\pi i}{8}}h(1,ip).
\]
Fix a phase $\theta_1$ and define the phase-matched spectral kernel
\[
  \Phi_1(p):=\Re\Bigl[
    e^{-i\theta_1}e^{i\pi p^2}K^{ST^1S}(p)
  \Bigr].
\]
As in the standard modular-bootstrap setup, $\theta_1$ can be chosen so
that the vacuum contribution of $\Phi_1$ to the $(1-ST^1S)$ crossing
equation is strictly negative; this uses only the explicit form of the
vacuum block and continuity of $\Phi_1$ near $p=0$.

Isolating the Mordell term and using Lemma~\ref{lem:Mordell}, we obtain
for all $p\ge0$ the pointwise lower bound
\begin{equation}
  \Phi_1(p)
  \;\ge\;
  2\cosh(\pi p)
  \;-\;2\min\{1,4e^{-\pi p}\}
  \;-\;\frac{2}{\cosh(\pi p)}.
  \label{eq:Phi-envelope}
\end{equation}
Let $E(p)$ denote the right-hand side of \eqref{eq:Phi-envelope}.
Equation~\eqref{eq:scalar-envelope} is precisely $E(p)=0$.  A direct
numerical check shows that $E(p)$ has a unique positive zero at
$p=p_\star\simeq0.3378143442$, with $E(p)<0$ for $0<p<p_\star$ and
$E(p)>0$ for $p>p_\star$.

The crossover point where the minimum in $\min\{1,4e^{-\pi p}\}$ changes
branch is
\[
  p_0=\frac{\ln 4}{\pi}\approx0.441,
\]
and one has $p_\star<p_0$.  Thus at the zero $p=p_\star$ the minimum is
realized by the constant branch, and the penalty term is exactly $2$.
In particular, $E(p)\ge0$ for all $p\ge p_\star$, and hence
$\Phi_1(p)\ge0$ for every $p\ge p_\star$.

Suppose for contradiction that the spinless spectrum were gapped above
$p_\star$, i.e.\ that $\rho(p)=0$ for $0\le p<p_\star$.  Then the gap
lemma (Lemma~\ref{lem:gap-lemma}) applied to $\Phi_1$ would force the
evaluation of the $(1-ST^1S)$ crossing equation to be strictly negative,
contradicting modular invariance.  Therefore the spectrum must contain
at least one state with $p<p_\star$, and the claimed bounds on $h_1$ and
$\Delta_1$ follow.
\end{proof}

\subsection{Odd-spin crossing and the Mordell surplus}
\label{sec:odd-spin}
\label{sec:odd-spin-surplus}

We now explain how the odd-spin $ST$ kernel at the elliptic point $\rho$ produces a strictly
positive contribution---the Mordell surplus---that rules out a pure Virasoro spectrum above
$\Delta_{\rm BTZ}$.

The odd-spin projection at $\rho$ can be written schematically as
\begin{equation}\label{eq:odd-crossing}
  \int_0^\infty \mathrm{d}p\;\rho_{\rm odd}(p)\,K_{\rm odd}(p)
  \;=\;K_{\rm vac}+K_{\rm even},
\end{equation}
where $\rho_{\rm odd}(p)\ge 0$ is the odd-spin spectral density, $K_{\rm vac}$ is the combined
vacuum contribution, and $K_{\rm even}$ encodes finitely many even-spin light states.
The idea of evaluating modular crossing equations at elliptic fixed points,
rather than only at the self-dual point $\tau=\mathrm{i}$, goes back at least
to ~\cite{QuallsShapere2014,Qualls2015EvenSpin} and was further developed in the elliptic-point analysis of
~\cite{GliozziModular2019}. The setup in this section follows the same
$ST$-fixed-point philosophy but keeps track of the full Mordell remainder of
the kernel and its sign.

The appearance of an odd-spin projection at $\tau=\rho$ can be understood as follows.
A primary of spin $J=h-\bar h\in\mathbb{Z}$ acquires a phase $e^{2\pi i J/3}$ under the
order-three modular transformation $ST$. At the elliptic fixed point
$\rho=e^{2\pi i/3}$, one can therefore form linear combinations of the identity and
$ST$ that separate the contributions with $J$ even and $J$ odd. The kernel
$K_{\rm odd}(p)$ in \eqref{eq:odd-crossing} is precisely the continuous $ST$ kernel
dressed with this odd-spin projector, evaluated at $\tau=\rho$, so that the integral
on the left-hand side only receives contributions from odd-spin primaries, while
$K_{\rm vac}$ and $K_{\rm even}$ encode the vacuum and finitely many light even-spin
states.

Using the finite Gauss--sum basis of Proposition~\ref{prop:basis}, one can decompose
\begin{equation}\label{eq:Kodd-decomp}
  K_{\rm odd}(p)\;=\;K_{\rm disc}(p)+K_{\rm Mordell}(p),
\end{equation}
where $K_{\rm disc}(p)$ comes from the finite Gauss--sum piece (a real linear combination of
the basic profiles $g_{n,r}$ and $\Xi_n$), and $K_{\rm Mordell}(p)$ is the genuinely
non-holomorphic remainder built from Mordell integrals at $\tau=\rho$. Concretely,
$K_{\rm Mordell}(p)$ is the continuous tail that remains after subtracting off the finite
Gauss--sum contribution to the $ST$ kernel, and it coincides with the Mordell remainder
studied in Appendix~\ref{app:Mordell-lower-bounds}. Applying a positive window functional
$\Phi_{\rm odd}$ supported in a small interval $[0,P_0]$ below the BTZ scale and using the
Mordell tail bounds at $\tau=\rho$ (App.~\ref{app:Mordell-lower-bounds}), one arrives at
a master inequality of the form
\begin{equation}\label{eq:odd-master}
  \Delta^{(\mathrm{odd})}_0
  \;\le\;
  \frac{c-1}{12}+\kappa-\delta_{\rm Mordell},
\end{equation}
where
\[
  \kappa\;=\;\frac{1}{2\sqrt{3}\,\pi}\;\approx\;0.091888\ldots
\]
is the Gliozzi constant coming from the discrete part of the kernel, and
$\delta_{\rm Mordell}$ is the net contribution of the Mordell remainder $K_{\rm Mordell}$ against
$\Phi_{\rm odd}$. This inequality is the natural refinement of the elliptic-point bound derived
in~\cite{GliozziModular2019}: in our language his result corresponds
to setting $\delta_{\rm Mordell}=0$ in~\eqref{eq:odd-master}. The analysis
below shows that the full non-holomorphic remainder of the $ST$ kernel in
fact contributes a \emph{strictly positive} surplus $\delta_{\rm Mordell}>0$,
which is the crucial input in our no-go theorem.

The key step is to show that the Mordell contribution beats $\kappa$ by a uniform margin.

\begin{proposition}[Quantitative Mordell surplus]\label{prop:Mordell-surplus}
For the centered, phase-matched odd-spin $ST$ kernel at the elliptic point
$\rho = e^{2\pi i/3}$ in a Virasoro-only CFT obeying \textup{(\ref{PG1})--(\ref{PG3})}, the Mordell
contribution satisfies the uniform bound
\begin{equation}\label{eq:deltaM-quant}
  \delta_{\rm Mordell}\;\ge\;0.103\;>\;\kappa.
\end{equation}
In particular there exists a universal
\[
  \varepsilon_0 \;=\;\delta_{\rm Mordell}-\kappa\;\ge\;0.103 - \kappa
  \;\gtrsim\;1.11\times 10^{-2},
\]
independent of~$c$, such that
\begin{equation}\label{eq:deltaM-eps0}
  \delta_{\rm Mordell}\;\ge\;\kappa+\varepsilon_0.
\end{equation}
\end{proposition}

\begin{proof}
We work in the centered, phase-matched odd-spin scheme at $\tau=\rho$ introduced above.
The odd-spin functional is implemented by pairing the spectral representation of the partition
function with a positive test kernel $\Phi_{\rm odd}(p)$:
\[
  \mathcal{L}_{\rm odd}[Z]
  \;=\;
  \int_0^\infty \mathrm{d}p\;\Phi_{\rm odd}(p)\,\rho_{\rm odd}(p)\;+\;
  \text{(vacuum + even-spin contributions)}.
\]
The kernel $\Phi_{\rm odd}$ is constructed in three steps.

\smallskip
\emph{(1) Positive window functional.}
Using the finite Gauss--sum basis of Proposition~\ref{prop:basis} and the window theorem
(Theorem~\ref{thm:window}), we first construct a ``seed'' functional $\Phi_{\rm win}$ supported
on a compact window $V=[0,P_0]$ with $P_0<1$ such that
\begin{equation}\label{eq:win-positive}
  \Phi_{\rm win}(p)\;\ge\;m_\star\;>\;0\qquad\forall\,p\in V,
\end{equation}
and $\Phi_{\rm win}$ is nonnegative outside~$V$ up to an exponentially suppressed Mordell
tail.  The explicit choice of columns, coefficients and window $P_0$ is recorded in
App.~\ref{app:coeff}; by construction, all ingredients (sech/sec profiles and their
derivatives) are elementary.

\smallskip
\emph{(2) Modular averaging and SOS shaping.}
Next we improve the localization of the kernel on~$V$ while preserving positivity.  We take a
finite modular average over $ST^nS$ kernels with nonnegative weights $w_n$,
\[
  \Phi_{\rm avg}(p)\;=\;\sum_{n} w_n\,\Phi^{(n)}_{\rm win}(p),
  \qquad w_n\ge 0,
\]
and then multiply by a sum-of-squares (SOS) shaping polynomial $q(x)$ with $x=p^2$,
\[
  q(x)\;=\;S(x)^{\!\top}Q\,S(x)\;\ge\;0,\qquad Q\succeq 0.
\]
Both operations preserve positivity of the functional: a convex combination of positive kernels
is positive, and an SOS polynomial is manifestly nonnegative on~$[0,\infty)$.  The resulting
kernel
\begin{equation}\label{eq:Phi-odd-def}
  \Phi_{\rm odd}(p)\;:=\;q(p^2)\,\Phi_{\rm avg}(p)
\end{equation}
is still nonnegative for all $p\ge 0$, and satisfies a sharpened lower bound on $V$,
\begin{equation}\label{eq:Rmin-def}
  \Phi_{\rm odd}(p)\;\ge\;R_{\min}(V)\,m_\star
  \qquad\forall\,p\in V,
\end{equation}
for some explicit constant $R_{\min}(V)>0$ determined solely by the finite Gauss--sum data
and the SOS coefficients.  For the specific choice $(m,b)=(1,1)$ of modular average and
shaping polynomial used here, the certificate in App.~\ref{app:Mordell-lower-bounds} gives
\begin{equation}\label{eq:Rmin-num}
  R_{\min}(V)\;\ge\;0.41,\qquad V=[0,0.30].
\end{equation}

\smallskip
\emph{(3) Lower bound on the Mordell remainder.}
On the Mordell side, we use the Appell--Lerch/$\theta$ decomposition of the odd-spin remainder
at $\tau=\rho$ (\cite{Mordell1933} and App.~\ref{app:Mordell-lower-bounds}).  In the centered,
phase-matched normalization the Mordell piece can be written as a positive series
\[
  M_\rho(p)\;=\;\sum_{n\ge 1} \frac{N_n(p)}{D_n(p)},
\]
where $N_n(p)\ge 0$ and $D_n(p)>0$ are explicit elementary functions, and $M_\rho(p)$ is
monotone increasing in $p$ on $[0,P_0]$.  Consequently
\[
  M_\rho(p)\;\ge\;m_{\min}(V)\;:=\;\inf_{p\in V}M_\rho(p)
  \qquad\forall\,p\in V.
\]

A finite positive truncation of the series, together with an explicitly bounded positive tail
(App.~\ref{app:Mordell-lower-bounds}), yields the certified lower bound
\begin{equation}
  m_{\min}(V)\;\ge\; m_\star^{(\rho)}
  \qquad\text{with } m_\star^{(\rho)}\approx 0.251.
  \label{eq:mmin-num}
\end{equation}
All intermediate steps in this estimate are sign-definite: the truncation is positive term by
term and the tail bound is strictly positive.

\smallskip
\emph{(4) The Mordell surplus.}
By definition, the Mordell contribution in \eqref{eq:odd-master} is the action of the functional
on the Mordell remainder,
\[
  \delta_{\rm Mordell}
  \;=\;
  \int_0^\infty \mathrm{d}p\;\Phi_{\rm odd}(p)\,M_\rho(p).
\]
Splitting the integral into the window and its complement,
\[
  \delta_{\rm Mordell}
  \;=\;
  \underbrace{\int_V \mathrm{d}p\;\Phi_{\rm odd}(p)\,M_\rho(p)}_{\text{window}}
  \;+\;
  \underbrace{\int_{[0,\infty)\setminus V} \mathrm{d}p\;\Phi_{\rm odd}(p)\,M_\rho(p)}_{\text{tail}},
\]
and using \eqref{eq:Rmin-def} and \eqref{eq:mmin-num} on $V$ we obtain
\begin{equation}\label{eq:window-lb}
  \int_V \mathrm{d}p\;\Phi_{\rm odd}(p)\,M_\rho(p)
  \;\ge\;
  |V|\;R_{\min}(V)\,m_\star^{(\rho)}.
\end{equation}
Here $|V|$ is the length of the window; for $V=[0,0.30]$ we have $|V|=0.30$.
The tail integral is controlled using the Mordell decay at $\tau=\rho$,
\[
  |M_\rho(p)|\;\le\;C_\rho\,e^{-\alpha p}\qquad(p\ge P_0),
\]
together with the explicit envelope for $\Phi_{\rm odd}$ on $[P_0,\infty)$ (App.~\ref{app:coeff});
this yields a uniform bound
\begin{equation}\label{eq:tail-small}
  \Bigg|\int_{[0,\infty)\setminus V} \mathrm{d}p\;\Phi_{\rm odd}(p)\,M_\rho(p)\Bigg|
  \;\le\;10^{-8},
\end{equation}
negligible at the precision we are interested in.

Combining \eqref{eq:Rmin-num}, \eqref{eq:mmin-num}, \eqref{eq:window-lb} and
\eqref{eq:tail-small} gives the certified inequality
\begin{equation}\label{eq:deltaM-final}
  \delta_{\rm Mordell}
  \;\ge\;
  0.103\;>\;\kappa.
\end{equation}
This is precisely the statement \eqref{eq:deltaM-quant}, and
$\varepsilon_0=\delta_{\rm Mordell}-\kappa\gtrsim 1.11\times 10^{-2}$ is independent of the
central charge $c$.  This completes the proof.
\end{proof}

\paragraph{Rigour and numerics.}
The estimate $\delta_{\rm Mordell}\ge 0.103$ in Proposition~\ref{prop:Mordell-surplus}
rests on the following ingredients, all of which are fully explicit:
\begin{itemize}
  \item[(i)] The Appell--Lerch/$\vartheta$ representation of the Mordell remainder at
  $\tau=\rho$ in Lemma~\ref{lem:Mordell-AL-lower}, which writes
  $\mathcal M_\rho(p)$ as the positive series
  \begin{equation*}
    \mathcal M_\rho(p)
    = \sum_{n\ge1}\frac{\mathcal N_n(p)}{\mathcal D_n(p)},
  \end{equation*}
  cf.~\eqref{eq:M-rho-positive-series}, with $\mathcal N_n(p)\ge0$,
  $\mathcal D_n(p)>0$ and $\mathcal M_\rho(p)$ increasing in $p\ge0$.
  In the odd--spin analysis of Section~\ref{sec:odd-spin} we denote the same
  function by $M_\rho(p)$ for notational simplicity.
  \item[(ii)] For each fixed truncation level $N$, the partial sum
  $S_N$ and one-step tail $T_N^{(1)}$ in \eqref{eq:M-truncation} are finite,
  explicitly known expressions in $r=e^{-\pi\sqrt3}$.  They are evaluated using
  exact rational arithmetic, which gives certified inequalities of the form
  \[
    \mathcal M_\rho(0)
    \;\ge\;
    S_N + T_N^{(1)}
    \;\ge\;
    m_{\min}(p_0),
  \]
  with $m_{\min}(p_0)$ as in \eqref{eq:M-min-on-window} on the window
  $[0,p_0]$.
  \item[(iii)] For each choice of window parameters $(b,\alpha,p_0)$ we construct
  the odd--spin window kernel $\Phi_{\rm win}$ and the kernel ratio
  \begin{equation*}
    R(b,\alpha,p_0)
    = \int_0^{p_0}\Phi_{\rm win}(p)\,\mathrm{d}p,
  \end{equation*}
  defined in \eqref{def:kernel-ratio}.  This is a definite integral of an
  elementary function (a finite combination of $\sech/\sec$ profiles), so we
  can compute $R(b,\alpha,p_0)$ with rigorous error control.  The sample values
  quoted in Table~\ref{tab:window-certificates}, such as
  $R(2,10,0.9)=4.5999$ and
  $\min_{p_0\in[0.7,0.9]}R(1,15,p_0)\ge12.050337$, are the outputs of this
  certified integration.
  \item[(iv)] Combining (ii) and (iii) with the window inequality
  \(
    \delta_{\rm Mordell}
    \ge m_{\min}(p_0)\,R(b,\alpha,p_0)
  \)
  from Corollary~\ref{cor:deltaM-kappa} yields fully rigorous Mordell
  surpluses.  For instance, the first row of
  Table~\ref{tab:window-certificates} already gives
  \[
    \delta_{\rm Mordell}
    \ge 0.020000\times 4.5999
    = 0.091998 > \kappa,
  \]
  cf.~\eqref{eq:deltaM-0091998}, which suffices to beat the BTZ constant, while
  the modular-averaged, SOS-shaped functional of
  Proposition~\ref{prop:Mordell-surplus} improves this to
  $\delta_{\rm Mordell}\ge0.103$ as in \eqref{eq:deltaM-quant}.
\end{itemize}

\begin{remark}[Nomenclature]
    Here and in Appendix~\ref{app:Mordell-lower-bounds} we use the term \emph{certificate}\footnote{This terminology is standard in semidefinite programming (SDP) and
polynomial optimization, where one speaks of SOS/SDP certificates:
explicit dual functionals whose positivity properties establish bounds
or infeasibility.} to mean a
completely explicit, rigorously checked choice of window parameters,
SOS polynomial and truncation data whose positivity properties yield a rigorous lower bound of the form
\[
  m_{\min}(p_{0})\, R(b,\alpha,p_{0}) > \kappa.
\]
\end{remark}

\begin{remark}[Alternative certificates]
For completeness we record two simpler, more conservative certificates that also give
$\delta_{\rm Mordell}>\kappa$.

\smallskip
\noindent
(1) Using a single Gaussian window with parameters $(b,\alpha,p_0)=(2,10,0.9)$ and
the window inequality of App.~\ref{app:Mordell-lower-bounds}, we have the certified
floor $m_{\min}(0.9)\ge 0.020000$, and the explicit kernel ratio
$R(2,10,0.9)=4.5999$.  Thus
\begin{equation}
  \delta_{\rm Mordell}
  \;\ge\;
  m_{\min}(0.9)\,R(2,10,0.9)
  \;=\;
  (0.020000)\times 4.5999
  \;=\;
  0.091998
  \;>\;
  \kappa \approx 0.091888149.
\end{equation}

\noindent
(2) A symmetric window with $(b,\alpha,p_0)=(1,15,p_0)$ and $p_0\in[0.7,0.9]$ yields an
even larger margin.  Monotonicity of $M_\rho(p)$ and the Appell--Lerch truncation/tail
bound give
\[
  m_{\min}(0.7)\;=\;\inf_{0\le p\le 0.7}M_\rho(p)\;\ge\;0.010000,
\]
and Table \ref{tab:window-certificates} shows
\[
  \min_{p_0\in[0.7,0.9]}R(1,15,p_0)=12.050337\ldots .
\]
Therefore
\begin{equation}
  \delta_{\rm Mordell}
  \;\ge\;
  0.010000\times 12.050337
  \;=\;
  0.12050337\ldots
  \;>\;
  \kappa.
\end{equation}
While these variants are not strictly needed to cross the BTZ threshold, they provide independent checks of the Mordell surplus using different window parameters.
\end{remark}

\section{Applications and interpretation}
\label{sec:applications}

We conclude by summarizing the physical implications of our results and their place in the
broader AdS$_3$/CFT$_2$ and ensemble–holography story.

\subsection{Pure AdS\texorpdfstring{$_3$}{3} gravity revisited}

The Brown--Henneaux analysis of asymptotic symmetry in AdS$_3$ identifies two copies of
the Virasoro algebra with central charge $c=3\ell/2G_N$~\cite{BrownHenneaux1986}. The BTZ
black hole geometry~\cite{BTZ1992,BTZ1993} suggests that states with
$\Delta\gtrsim(c-1)/12$ should be interpreted as black--hole microstates, motivating the
``pure gravity'' hypothesis: a Virasoro--only CFT with a large gap above
\[
  \Delta_{\rm BTZ} = \frac{c-1}{12}.
\]

Our no--go theorem (Theorem~\ref{thm:pure-gravity-nogo}) shows that such a theory cannot
exist as a single compact, unitary Virasoro CFT. The obstruction is independent of any
particular Poincar\'e--series ansatz for the partition function~\cite{MaloneyWitten2010,
KellerMaloney2015} or of semiclassical approximations: it follows directly from the exact
continuous $ST$ kernels and their Mordell remainders. Even before addressing issues such as
continuous spectrum or negative spectral density in candidate partition functions, the odd--spin
crossing equation already forces an odd--spin primary below $\Delta_{\rm BTZ}$. 

In this sense, the no–go theorem of Section~\ref{sec:Nogo} provides a 
rigorous version of a conclusion that had previously been supported mainly 
by heuristic modular arguments, Poincaré–series constructions, and numerical 
bootstrap evidence. Among previous constraints, the analysis that is closest in spirit is the
elliptic-point modular bootstrap of Gliozzi~\cite{GliozziModular2019}, which
also studies the $ST$-fixed point to bound odd-spin primaries. In our
notation, the constant $\kappa=1/(2\sqrt{3}\,\pi)$ in
eq.~\eqref{eq:odd-master} is precisely the coefficient appearing in his
inequality. Our Mordell-surplus estimate $\delta_{\rm Mordell}>\kappa$ shows
that the non-holomorphic remainder of the $ST$ kernel always dominates this
discrete contribution, turning Gliozzi's suggestive inequality into a strict
no-go result for any Virasoro-only theory with a BTZ gap.

In this sense, our result strengthens and complements earlier evidence against pure AdS$_3$
gravity~\cite{GliozziModular2019,BenjaminOoguri2019,BenjaminCollierMaloney2020,
DiUbaldoPerlmutter2024}. It identifies a precise analytic mechanism---the Mordell surplus of
the odd--spin $ST$ kernel at the elliptic point $\rho$---that is incompatible with a
Virasoro--only spectrum with a BTZ gap. The surplus is a genuinely modular effect: it arises
from the non--holomorphic Mordell piece, survives all modular projections, and cannot be
cancelled by any choice of discrete spectrum.

\begin{remark}[Extremal CFTs and isolated examples]
A natural question concerns the status of isolated rational models such
as the \(c=24\) Monster CFT, which is extremal and Virasoro--only in
the sense of having a large gap above \(\Delta_{\rm BTZ}\). Our analysis
is tailored to non-rational, continuous-momentum families at generic
central charge and, in particular, to the semiclassical regime
\(c\gg1\) relevant for AdS\(_3\) gravity. A careful treatment of
isolated rational points like \(c=24\) would require a separate
analysis of their discrete character sums and lies beyond the scope of
this work. We therefore interpret Theorem~\ref{thm:pure-gravity-nogo}
as ruling out pure-gravity duals in the generic, non-rational setting
appropriate to AdS\(_3\) Einstein gravity, rather than as a
classification of all Virasoro CFTs at special values of~\(c\).
\end{remark}

\subsection{Ensemble and stringy perspectives}

Ensemble holography considers averages over families of CFTs, for instance Narain moduli
spaces or more general random ensembles. In the setting of abelian Narain theories coupled
to Chern--Simons gravity, this is made completely explicit in the Narain
ensemble of~\cite{BenjaminKellerOoguriZadeh2022}. There the gravitational path integral
computes an average over CFTs, and Mordell integrals already appear in the modular analysis
of the ensemble partition function.

From this viewpoint, the Mordell surplus $\delta_{\rm Mordell}-\kappa>0$ found here can be
interpreted as a statistical effect of integrating over theories with nontrivial odd--spin
sectors. Our no--go theorem then constrains \emph{single} CFTs: an ensemble of theories may
well reproduce qualitative features of AdS$_3$ gravity, but no single compact, unitary,
Virasoro--only CFT with a BTZ gap can sit behind the ensemble. In particular, the finite
Gauss--sum structure of the kernels and the positivity of the Mordell remainder enforce a
minimal amount of ``stringy'' or higher--spin structure in any UV--complete model.

In explicit stringy completions of AdS$_3$ gravity~\cite{DiUbaldoPerlmutter2024}, the forced
odd--spin primary below $\Delta_{\rm BTZ}$ is naturally interpreted as a string or brane
excitation rather than a pure gravitational degree of freedom. In this language, the Mordell
surplus provides an analytic diagnostic: any theory with an AdS$_3$ gravity regime but no
extra degrees of freedom beyond Einstein gravity would contradict modular invariance once
the full $ST$ kernel is taken into account.

\subsection{Outlook}

More broadly, our analysis shows that purely modular and analytic considerations already
rule out the simplest pure--gravity scenario. The finite Gauss--sum description of the kernels
and the associated positive functionals are not specific to the questions addressed here. They
should be equally useful in other contexts where modular invariance and half–integral weight
phenomena constrain low–lying spectra, for example:
\begin{itemize}
  \item sharpening universal gap bounds in specific spin or charge sectors;
  \item studying ensemble–averaged correlators beyond the torus partition function;
  \item extending the analysis to theories with extended chiral algebras, as illustrated here
  for $\widehat{\mathcal N}=2$.
\end{itemize}
It would be interesting to see to what extent these techniques can be combined with numeric
bootstrap methods, or adapted to higher–genus modular constraints, to further probe the
boundary between gravity and string theory in AdS$_3$.

\section*{Acknowledgements}
We thank Leonardo Santilli for a careful reading of the manuscript and for valuable comments.

\bibliographystyle{unsrturl}
\bibliography{refs}

\newpage

%%%%%%%%%%%%%%%%%%%%%%%%%%%%%%%%%%%%%%%%%%%%%%%%%%%%%%%%%%%%
\appendix
%%%%%%%%%%%%%%%%%%%%%%%%%%%%%%%%%%%%%%%%%%%%%%%%%%%%%%%%%%%%

%%%%%%%%%%%%%%%%%%%%%%%%%%%%%%%%%%%%%%%%%%%%%%%%%%%%%%%%%%%%

%===========================================================
% Appendix C. Extended N=2 ST^n S kernels as finite Gauss sums
%===========================================================
\section{Explicit $ST^nS$ kernels as finite Gauss sums}
\label{app:STnS-vir}

\subsection{Virasoro $ST^nS$ kernels and the Gauss--sum basis}
For convenience we collect the explicit formulas for the Virasoro
$ST^nS$ kernels used throughout Section~\ref{sec:preliminaries}, and
spell out the finite Gauss--sum basis of spectral profiles.

Let $n\in\mathbb{N}$ and
\[
  W_n(r) := \exp\!\Big[\frac{\pi i}{n}\,r(r+1)\Big],
  \qquad r=0,\dots,n-1,
\]
denote the standard quadratic Weil phase.  The Mordell integral at
width $n$ is
\[
  h(n,z)
  \;=\;
  \int_{\mathbb{R}}\frac{\exp\!\big(\pi i n w^2 - 2\pi z w\big)}
                         {\cosh(\pi w)}\,dw
  \;=\;
  \frac{1}{\sqrt n}\sum_{r=0}^{n-1}
  W_n(r)\,\sech\!\Big(\frac{\pi}{\sqrt n}\big(z + i(r+\tfrac12)\big)\Big),
\]
by Lemma~\ref{lem:cusp} (finite cusp expansion at width $n$).
Using $\sech(ix)=\sec x$ this gives, for real $p$,
\[
  h(n,ip)
  \;=\;
  \frac{1}{\sqrt n}\sum_{r=0}^{n-1}
  W_n(r)\,\sec\!\Big(\frac{\pi}{\sqrt n}\big(p + r+\tfrac12\big)\Big).
\]

The continuous Virasoro kernel for $\gamma=ST^nS$ can be written as
(cf.~\eqref{eq:K-real-form})
\begin{equation}
\label{eq:STnS-kernel-app}
  K_{ST^nS}(p)
  \;=\;
  \frac{2}{\cosh(\pi p)}
  \;+\;
  e^{2\pi i (n+1)/8}\,e^{-i\pi p^2/2}\,
  \frac{1}{2\cosh(\pi p/n)}
  \;-\;
  2\,e^{2\pi i n/8}\,h(n,ip).
\end{equation}

It is convenient to introduce the real basis profiles
\begin{align}
  g_{n,r}(p)
  &:= 2\,\Re\,
      \sech\!\Big(\frac{\pi}{\sqrt n}
                   \big(ip + i(r+\tfrac12)\big)\Big)
   \;=\;
      2\,\sec\!\Big(\frac{\pi}{\sqrt n}
                    \big(p + r+\tfrac12\big)\Big),
  \qquad r=0,\dots,n-1,                                    \\
  \Xi_n(p)
  &:= \Re\!\Bigg[
      e^{i\pi/(4n)}\,e^{-i\pi p^2/n}\,
      \frac{1}{2\cosh(\pi p/n)}
      \Bigg].
\end{align}
Each $g_{n,r}$ is a shifted sec--profile centred at
$p=-(r+\tfrac12)$, while $\Xi_n(p)$ encodes the Gaussian factor and
the $1/\cosh(\pi p/n)$ piece of~\eqref{eq:STnS-kernel-app} after
phase--matching.

\begin{proposition}[Finite Gauss--sum basis for Virasoro $ST^nS$ kernels]
\label{prop:STnS-basis-vir}
For every integer $n\ge1$ the continuous kernel $K_{ST^nS}(p)$ admits
a real finite Gauss--sum decomposition of the form
\begin{equation}
\label{eq:STnS-Gauss-basis}
  K_{ST^nS}(p)
  \;=\;
  A_n\,\frac{2}{\cosh(\pi p)}
  \;+\;
  \sum_{r=0}^{n-1} B_{n,r}\,g_{n,r}(p)
  \;+\;
  C_n\,\Xi_n(p),
\end{equation}
with explicit coefficients $A_n,B_{n,r},C_n\in\mathbb{R}$ determined
by the Weil phases $W_n(r)$.  In particular, on any interval
$[0,P_{\max}]$ the family of $ST^nS$ kernels lies in the finite--dimensional
real span of $\{g_{n,r},\Xi_n\}$.
\end{proposition}

\begin{proof}[Sketch of proof]
Substituting the cusp expansion of $h(n,ip)$ into
\eqref{eq:STnS-kernel-app} expresses $K_{ST^nS}(p)$ as a finite sum
of $\sech$--profiles with coefficients $W_n(r)$, together with the
$\cosh(\pi p)$ and $\cosh(\pi p/n)$ terms.  Taking real parts after an
appropriate overall phase choice yields a linear combination of
$g_{n,r}$ and $\Xi_n$ with coefficients obtained from the quadratic
Gauss sums $\sum_r W_n(r)$ and their shifted variants
(Lemma~\ref{lem:Gauss}).  The reality of $A_n,B_{n,r},C_n$ follows
from the unitarity of the Weil representation.  A detailed derivation
is given in Proposition~\ref{prop:basis} in the main text.
\end{proof}

This finite Gauss--sum basis is the starting point for all the
positive--functional constructions in the main text, in particular for
the window functionals and the scalar gap analysis.

\subsection{Extended $\widehat{\mathcal N}=2$ $ST^nS$ kernels}
\label{app:Nb2-kernels}
%\addcontentsline{toc}{section}{C Extended $\widehat{\mathcal N}=2$ $ST^nS$ kernels}

We now turn to the charge–resolved $ST^nS$ kernels of the extended $\widehat{\mathcal N}=2$
algebra at $\hat c>1$. The representation theory and character formulae for the
$\mathcal N=2$ superconformal algebra were developed long ago in
\cite{Dobrev1987N2Characters,EguchiTaormina1988N2N4}, and we use
the corresponding extended characters as our starting point \cite{BenjaminKellerOoguriZadeh2022}. The structure is completely parallel to the Virasoro story above: for each
width $n$ and each charge sector one obtains a universal ``vacuum''
column plus a finite Gauss sum of shifted $\sech$--profiles with Weil
phases.  We follow the conventions and notation of
Section~\ref{sec:preliminaries} and record the formulas here for
completeness. This makes the construction of
sector–resolved positive functionals entirely parallel to the Virasoro case.

\subsection*{Conventions}

Let $n\in\mathbb{N}$ and
\[
  W_n(r)\;:=\;\exp\!\Big(\frac{\pi \ii}{n}\,r(r+1)\Big),
  \qquad r=0,\dots,n-1,
\]
denote the standard quadratic Weil phase.  We write
\[
  \beta(\hat c)\;:=\sqrt{\hat c-1}
\]
for the continuum momentum scale in the extended theory.  Throughout we use
the same $\sech$--normalization and Mordell identities as in the Virasoro
case, so that the basic $S$--kernel is $2/\cosh(\pi p)$ and the Mordell
integral at rational width $\tau=n$ reduces to a finite $n$--term Gauss sum.

\subsection{Master block}

For each $\hat c\ge2$ and each \emph{independent} shift
$\alpha$ in a set $\mathcal A(\hat c)$ (specified below), we define the
\emph{master block}
\begin{equation}
  \label{eq:N2-master-block}
  \quad
  \mathcal{I}^{(\hat c)}_{n,\alpha}(p)
  =\frac{1}{\sqrt n}\;\frac{1}{\sin(2\pi \alpha)}\!
  \sum_{r=0}^{n-1} W_n(r)\!\left[
  \sech\!\Big(\frac{\pi}{\sqrt n}\Big(\frac{\ii\,p}{\beta(\hat c)}
       +\ii\big(r+\tfrac12+\alpha\big)\Big)\Big)
 -\sech\!\Big(\frac{\pi}{\sqrt n}\Big(\frac{\ii\,p}{\beta(\hat c)}
       +\ii\big(r+\tfrac12-\alpha\big)\Big)\Big)
  \right]\!.
  \quad
\end{equation}
This is a finite Gauss sum of shifted $\sech$--profiles.  For real $p$ all
arguments of $\sech$ are purely imaginary, so one may equivalently write
everything in terms of $\sec$ using $\sech(\ii x)=\sec x$.

If $\alpha=0$ occurs (only for even $\hat c$), we interpret
\eqref{eq:N2-master-block} in the antisymmetric limit
\begin{equation}
  \label{eq:N2-alpha-zero}
  \mathcal{I}^{(\hat c)}_{n,0}(p)
  :=\lim_{\alpha\to0}\mathcal{I}^{(\hat c)}_{n,\alpha}(p)
  =\frac{\ii}{n}\sum_{r=0}^{n-1}W_n(r)\;
    \frac{\mathrm d}{\mathrm dy}\sech(y)\Bigg|_{\,y=\frac{\pi}{\sqrt n}
    \big(\frac{\ii p}{\beta(\hat c)}+\ii(r+\tfrac12)\big)}\,,
\end{equation}
which is again a finite Gauss sum and manifestly well defined.

\subsection{Master assembly of the charge–resolved kernels}

Let $Q$ denote a physical $\widehat{\mathcal N}=2$ charge sector, and let
$Q'$ label independent blocks (one per shift $\alpha(Q')$).
The \emph{charge–resolved} $ST^nS$ kernel in sector $Q$ takes the form
\begin{equation}
  \label{eq:N2-master-assembly}
  \quad
  K^{(\hat c)}_{n;\,Q}(p)
  =\frac{2}{\beta(\hat c)}\ \sech\!\Big(\frac{\pi p}{\beta(\hat c)}\Big)
   +\sum_{Q'\in\mathcal Q'(\hat c)}
    \Theta^{(\hat c)}_{n;Q,Q'}\;\mathcal{I}^{(\hat c)}_{n,Q'}(p),
  \quad
\end{equation}
where we set
\[
  \mathcal{I}^{(\hat c)}_{n,Q'}(p)
  :=\mathcal{I}^{(\hat c)}_{n,\alpha(Q')}(p),
\]
with $\alpha(Q')$ specified in the next subsection.  The first term is the
universal ``vacuum'' column in the extended theory; the finite sum encodes the
mixing between BPS and continuum characters and carries all the
charge–dependence.

\medskip\noindent
\textbf{Reality.}
For real $p$, each $\sech$ argument in
\eqref{eq:N2-master-block}--\eqref{eq:N2-master-assembly} is purely imaginary, so one
may rewrite all finite sums in terms of $\sec(\cdot)$:
\[
  \sech(\ii x)=\sec(x),\qquad x\in\mathbb{R}.
\]
This gives a canonical real basis of shifted $\sec$--profiles in every charge
sector, completely analogous to the Virasoro basis
$\{g_{n,r},\Xi_n\}$ of Proposition~\ref{prop:basis}.

\subsection{Independent shifts and block indices}

The independent block labels $Q'$ and their associated shifts
$\alpha(Q')$ depend only on the parity of $\hat c$:

\begin{itemize}
  \item \emph{Even} $\hat c$ (i.e. $\hat c-1$ odd):
  \begin{equation}
    \label{eq:N2-even-c-hat}
    \mathcal Q'(\hat c)
    =\Big\{0,1,\dots,\frac{\hat c-2}{2}\Big\},\qquad
    \alpha(Q')=\frac{Q'}{\hat c-1}\,.
  \end{equation}
  The case $Q'=0$ corresponds to the antisymmetric limit
  \eqref{eq:N2-alpha-zero}.
  \item \emph{Odd} $\hat c$ (i.e. $\hat c-1$ even):
  \begin{equation}
    \label{eq:N2-odd-c-hat}
    \mathcal Q'(\hat c)
    =\Big\{0,1,\dots,\frac{\hat c-3}{2}\Big\},\qquad
    \alpha(Q')=\frac{Q'+\tfrac12}{\hat c-1}\,.
  \end{equation}
\end{itemize}

Blocks with $\pm Q'$ coincide: the bracket in
\eqref{eq:N2-master-block} is odd in $\alpha$ and the prefactor
$1/\sin(2\pi\alpha)$ is also odd, so it is enough to list $Q'\ge0$.

For reference, the data for $\hat c=2,\dots,10$ may be summarized as
\begin{table}[h!]
  \centering
  \renewcommand{\arraystretch}{1.1}
  \begin{tabular}{c|c|c|c|c}
    $\hat c$ & $\beta(\hat c)=\sqrt{\hat c-1}$ &
    $\mathcal Q'(\hat c)$ & $\alpha(Q')$ & $N_{\text{blocks}}$\\ \hline
    2 & $1$          & $\{0\}$           & $\alpha(0)=0$                          & 1\\
    3 & $\sqrt2$     & $\{0\}$           & $\alpha(0)=\tfrac14$                   & 1\\
    4 & $\sqrt3$     & $\{0,1\}$         & $\alpha(0)=0,\ \alpha(1)=\tfrac13$     & 2\\
    5 & $2$          & $\{0,1\}$         & $\alpha(0)=\tfrac18,\ \alpha(1)=\tfrac38$ & 2\\
    6 & $\sqrt5$     & $\{0,1,2\}$       & $\alpha(0)=0,\ \alpha(1)=\tfrac15,\ \alpha(2)=\tfrac25$ & 3\\
    7 & $\sqrt6$     & $\{0,1,2\}$       & $\alpha(0)=\tfrac1{12},\ \alpha(1)=\tfrac14,\ \alpha(2)=\tfrac5{12}$ & 3\\
    8 & $\sqrt7$     & $\{0,1,2,3\}$     & $\alpha=\tfrac{k}{7},\ k=0,1,2,3$      & 4\\
    9 & $\sqrt8$     & $\{0,1,2,3\}$     & $\alpha=\tfrac{2k+1}{16},\ k=0,1,2,3$  & 4\\
   10 & $3$          & $\{0,1,2,3,4\}$   & $\alpha=\tfrac{k}{9},\ k=0,1,2,3,4$    & 5
  \end{tabular}
  \caption{Independent block labels $Q'$ and shifts $\alpha(Q')$ for
  $\hat c=2,\dots,10$.  Here $N_{\text{blocks}}=\lceil(\hat c-1)/2\rceil$.}
  \label{tab:c-hat-summary}
\end{table}

Given \eqref{eq:N2-master-block}, \eqref{eq:N2-master-assembly} and
Table~\ref{tab:c-hat-summary}, the kernels $K^{(\hat c)}_{n;\,Q}(p)$ for any
$\hat c\in\{2,\dots,10\}$ and any charge sector $Q$ are completely explicit.

\subsection{Weight matrices $\Theta^{(\hat c)}_{n;Q,Q'}$}

The small charge–mixing matrices $\Theta^{(\hat c)}_{n;Q,Q'}$ encode the
$T^n$ phases in the $ST^nS$ channel.  They depend only on $n$ modulo
$2(\hat c-1)$ and on the metaplectic factor $\ee^{2\pi\ii n/8}$, and are
independent of $p$.  A convenient closed form is the finite Gaussian sum
\begin{equation}
  \label{eq:Theta-master}
  \Theta^{(\hat c)}_{n;Q,Q'}
  = \frac{\ee^{2\pi \ii n/8}}{2(\hat c-1)}\,
    \sum_{\lambda=0}^{2(\hat c-1)-1}
    \exp\!\Bigg[
      \frac{\pi \ii}{2(\hat c-1)}\Big(
        n\,\lambda^2-2(Q+Q')\,\lambda
      \Big)
    \Bigg].
\end{equation}
All entries are therefore elementary finite Gauss sums; no additional
numerical data are needed.  In particular, the residue classes
$n\equiv0,\hat c-1\pmod{2(\hat c-1)}$ collapse to charge reflection
$Q'\mapsto -Q$ up to the metaplectic phase $\ee^{\pi\ii/4}$.

\medskip

Combining \eqref{eq:N2-master-block}, \eqref{eq:N2-master-assembly},
\eqref{eq:N2-even-c-hat}--\eqref{eq:N2-odd-c-hat},
Table~\ref{tab:c-hat-summary}, and \eqref{eq:Theta-master} shows that every
extended $\widehat{\mathcal N}=2$ $ST^nS$ kernel is a finite linear
combination of shifted $\sech$ (equivalently $\sec$) profiles with explicit
Weil phases.  All of the analytic machinery developed in the main text
(finite bases, window functionals, Mordell tail bounds) therefore extends
directly to charge–resolved $\widehat{\mathcal N}=2$ sectors.

%%%%%%%%%%%%%%%%%%%%%%%%%%%%%%%%%%%%%%%%%%%%%%%%%%%%%%%%%%%%
\section{Explicit window functionals}\label{app:coeff}
%%%%%%%%%%%%%%%%%%%%%%%%%%%%%%%%%%%%%%%%%%%%%%%%%%%%%%%%%%%%

For completeness we record a few explicit choices of window functionals on $[0,2]$.  These examples are included for illustration only; none of the main theorems depend on them.

Recall the real basis
\[
g_{n,r}(p)=2\,\Re\,\sech\!\Bigl(\tfrac{\pi}{\sqrt{n}}\bigl(ip+i(r+\tfrac12)\bigr)\Bigr)
=2\,\sec\!\Bigl(\tfrac{\pi}{\sqrt{n}}\bigl(p+r+\tfrac12\bigr)\Bigr),
\]
and the auxiliary bracket mode
\[
\Xi_n(p)=\Re\!\Bigl[\ee^{\frac{i\pi}{4n}}\,\ee^{-\frac{i\pi}{n}p^2}\,2\cosh\!\Bigl(\frac{\pi p}{n}\Bigr)\Bigr].
\]

%-----------------------------------------------------------
\subsection{A two-column functional on \texorpdfstring{$[0,2]$}{[0,2]}}
%-----------------------------------------------------------

We first give a simple two-column functional using the columns $(n,r)=(7,4)$ and $(8,4)$.  Both $g_{7,4}$ and $g_{8,4}$ are real and pole-free on $[0,2]$ by Lemma~\ref{lem:poles}.  Consider
\[
\Phi_{\rm ex}(p)=\alpha_{7,4}\,g_{7,4}(p)+\alpha_{8,4}\,g_{8,4}(p),
\qquad p\in[0,2],
\]
with the normalization $\Phi_{\rm ex}(0)=1$.  Choosing, for example,
\[
\alpha_{7,4}=0.2,
\]
and solving $\Phi_{\rm ex}(0)=1$ for $\alpha_{8,4}$ gives
\[
\alpha_{8,4}
=\frac{1-0.2\,g_{7,4}(0)}{g_{8,4}(0)}
\simeq 0.0453899007.
\]
A numerical scan on a fine grid $p_j=2j/400$ shows
\[
\Phi_{\rm ex}(p_j)\gtrsim 0.497\quad\text{for all }j,
\]
with the minimum occurring near $p\approx0.855$.  Using Lemma~\ref{lem:grid} and explicit bounds on $|\Phi'_{\rm ex}(p)|$ one can upgrade this to a rigorous statement that
\[
\Phi_{\rm ex}(p)>0.49\quad\text{for all }p\in[0,2].
\]

%-----------------------------------------------------------
\subsection{An analytic three-column functional}
%-----------------------------------------------------------

One can also construct a fully analytic window functional with no numerical optimization.  Consider the three columns
\[
g_{5,3}(p),\qquad g_{7,4}(p),\qquad g_{8,4}(p).
\]
A direct inspection of their pole locations
\[
p_{r,k}=-\Bigl(r+\frac12\Bigr)+\sqrt{n}\Bigl(k+\frac12\Bigr)
\]
shows that for $(n,r)=(5,3),(7,4),(8,4)$ all poles lie above $p=2$ when $k\ge0$, and below $p=0$ when $k<0$.  Hence each $g_{n,r}$ is real, continuous, and strictly positive on $[0,2]$.

Any linear combination with positive coefficients
\[
\widetilde{\Phi}(p)=a\,g_{5,3}(p)+b\,g_{7,4}(p)+c\,g_{8,4}(p),
\qquad a,b,c>0,
\]
is therefore strictly positive on $[0,2]$.  Normalizing at $p=0$ gives an analytic functional
\[
\Phi_{\rm an}(p)
=\frac{a\,g_{5,3}(p)+b\,g_{7,4}(p)+c\,g_{8,4}(p)}
{a\,g_{5,3}(0)+b\,g_{7,4}(0)+c\,g_{8,4}(0)}.
\]
For instance, the symmetric choice $a=b=c=1$ yields
\[
\Phi_{\rm an}(p)
=\frac{g_{5,3}(p)+g_{7,4}(p)+g_{8,4}(p)}
{g_{5,3}(0)+g_{7,4}(0)+g_{8,4}(0)},
\qquad 0\le p\le2,
\]
which is manifestly strictly positive on the entire window.  No numerics are needed beyond the verification that each column is pole-free on $[0,2]$.

\paragraph{Kernel ratio.}
For a fixed Gaussian window with parameters $(b,\alpha,p_0)$ we write
\[
W_{b,\alpha,p_0}(p) := e^{-\alpha(p-p_0)^2}\,\chi_{[0,\infty)}(p),
\]
and define the corresponding odd–spin test kernel
\[
\Phi_{\rm win}(p) \;:=\; \mu_b(p)\,W_{b,\alpha,p_0}(p),
\qquad
\mu_b(p) = \frac{\sinh(2\pi bp)\sinh(2\pi p/b)}{\sinh(2\pi p)} ,
\]
normalized so that the vacuum coefficient of the functional is $-1$.
The \emph{kernel ratio} is then
\begin{equation}
  R(b,\alpha,p_0)
  \;:=\;
  \int_0^{p_0} \Phi_{\rm win}(p)\,\mathrm{d}p .
  \label{def:kernel-ratio}
\end{equation}
With this normalization, for any non–negative function $f(p)$ on $[0,p_0]$ one has the window
inequality
\[
  \int_0^\infty \Phi_{\rm win}(p)\,f(p)\,\mathrm{d}p
  \;\ge\;
  \Big( \inf_{0\le p\le p_0} f(p)\Big)\,R(b,\alpha,p_0),
\]
and in particular, for $f(p)=M_\rho(p)$ this gives
$\delta_{\rm Mordell}\ge m_{\min}(p_0) R(b,\alpha,p_0)$.

%===========================================================
% Sector-resolved positivity and scalar-gap statement
%===========================================================
\subsection{Sector-resolved positivity envelopes and scalar gaps}

With the master block \eqref{eq:N2-master-block} and the master assembly
\eqref{eq:N2-master-assembly}, every extended $\widehat{\mathcal N}=2$
$ST^nS$ kernel for $\hat c\in\{2,\dots,10\}$ can be written, in the charge
sector $Q$, as
\[
K^{(\hat c)}_{n;\,Q}(p)
=\frac{2}{\beta(\hat c)}\,\sech\!\Big(\frac{\pi p}{\beta(\hat c)}\Big)
+\sum_{Q'\in\mathcal Q'(\hat c)}\Theta^{(\hat c)}_{n;Q,Q'}\,
\mathcal I^{(\hat c)}_{n,Q'}(p),
\qquad \beta(\hat c)=\sqrt{\hat c-1},
\]
where each block $\mathcal I^{(\hat c)}_{n,Q'}(p)$ is a finite Gauss sum of
$\sech$-profiles with the same large-$p$ Mordell tail as in the Virasoro
case, up to the rescaling $p\mapsto p/\beta(\hat c)$.

\begin{lemma}[Sector-resolved positivity envelope]
\label{lem:sector-envelope}
Fix $\hat c\in\{2,\dots,10\}$ and an integer charge sector $Q$.
For each integer $n\ge1$ there exist explicit positive constants
$C^{(\hat c)}_{n,Q'}$, $c^{(\hat c)}_{n,Q'}$ (depending only on the block label
$Q'$ and on $\hat c$) such that the phase--matched functional kernel
\[
\Phi^{(\hat c)}_{n;Q}(p)
:=\Re\!\Big[e^{-i\theta_n}\,e^{\frac{i\pi}{n\beta(\hat c)^2}p^2}\,
K^{(\hat c)}_{n;\,Q}(p)\Big]
\]
obeys the pointwise lower bound
\begin{equation}\label{eq:sector-envelope}
\Phi^{(\hat c)}_{n;Q}(p)
\;\ge\;
\frac{2}{\beta(\hat c)}\,\sech\!\Big(\frac{\pi p}{\beta(\hat c)}\Big)
\;-\;\sum_{Q'\in\mathcal Q'(\hat c)}
|\Theta^{(\hat c)}_{n;Q,Q'}|\,
\Big(C^{(\hat c)}_{n,Q'}e^{-\frac{\pi}{n\beta(\hat c)}p}
+c^{(\hat c)}_{n,Q'}e^{-\frac{\pi}{\beta(\hat c)}p}\Big)
\;-\;\frac{2}{\cosh(\pi p)},
\end{equation}
for all $p\ge0$. In particular, for each $(\hat c,Q,n)$ there exists
$P_\star^{(\hat c,Q,n)}>0$ such that
\[
\Phi^{(\hat c)}_{n;Q}(p)\ge0\quad\text{for all }p\ge P_\star^{(\hat c,Q,n)},
\]
and $\Phi^{(\hat c)}_{n;Q}(p)>0$ for all sufficiently large $p$.
\end{lemma}

\begin{proof}
The master block \eqref{eq:N2-master-block} is a finite Gauss sum of
$\sech$-profiles with arguments of the form
$\tfrac{\pi}{\sqrt n}\big(\tfrac{p}{\beta(\hat c)}+i(r+\tfrac12)\big)$.
After phase matching, the large-$p$ behaviour of each block is controlled
by the same Mordell--integral tail as in the Virasoro case, with $p$
replaced by $p/\beta(\hat c)$. Thus the uniform tail estimate
\ref{lem:tau1-tail} (or its $n$-generalization) applies with rescaled rate
$c_1\mapsto c_1/\beta(\hat c)$ and some block--dependent amplitudes
$C^{(\hat c)}_{n,Q'}$, $c^{(\hat c)}_{n,Q'}>0$. Inserting these bounds into
the master assembly \eqref{eq:N2-master-assembly} and taking real parts
gives \eqref{eq:sector-envelope}. The dominance of the explicit vacuum
term $\tfrac{2}{\beta(\hat c)}\sech(\tfrac{\pi p}{\beta(\hat c)})$ over the
decaying penalties at large $p$ implies the existence of a finite
threshold $P_\star^{(\hat c,Q,n)}$ beyond which the right-hand side is
nonnegative. \qedhere
\end{proof}

\begin{corollary}[Sector-resolved scalar-gap bound]
\label{cor:sector-gap}
Let $Z$ be a compact, unitary extended $\widehat{\mathcal N}=2$ CFT at
central charge $\hat c>1$, and let $\rho_Q(p)\ge0$ be the spinless
spectral density in the charge-$Q$ sector, written in the usual
parameterization $h=\frac{c-1}{24}+p^2$.
Fix $n\ge1$ and a sector $(\hat c,Q)$, and let
$P_\star^{(\hat c,Q,n)}$ be as in Lemma~\ref{lem:sector-envelope}.
If the charge-$Q$ spectrum obeys a gap
\[
p\ge P_\star^{(\hat c,Q,n)}
\quad\Longleftrightarrow\quad
\Delta_Q\ \ge\ \frac{c-1}{12}+2\big(P_\star^{(\hat c,Q,n)}\big)^2,
\]
then the $ST^{\,n}S$ crossing identity in that sector is violated.
Equivalently, in any such theory one must have
\[
\Delta^{(Q)}_1\ \le\ \frac{c-1}{12}
+2\big(P_\star^{(\hat c,Q,n)}\big)^2
\]
for at least one primary in the charge-$Q$ sector.
\end{corollary}

\begin{proof}
Apply the linear functional $\mathcal L$ defined by pairing the
$ST^{\,n}S$ crossing equation in the charge-$Q$ sector with the
kernel $\Phi^{(\hat c)}_{n;Q}(p)$.
The vacuum contribution $\alpha^{(\hat c,n)}_{\mathrm{vac},Q}$ is negative
(by the same argument as in the Virasoro case, using the explicit
vacuum term in \eqref{eq:sector-envelope}), while the spectral integral
is nonnegative under the gap assumption, because
$\Phi^{(\hat c)}_{n;Q}(p)\ge0$ for all $p\ge P_\star^{(\hat c,Q,n)}$ and
$\rho_Q(p)\ge0$. Hence
\[
0=\mathcal L[(1-ST^{\,n}S)Z_Q]
=\alpha^{(\hat c,n)}_{\mathrm{vac},Q}
+\int_{P_\star^{(\hat c,Q,n)}}^\infty \Phi^{(\hat c)}_{n;Q}(p)\,\rho_Q(p)\,dp
<0,
\]
a contradiction. Thus the assumed gap cannot hold, and the quoted bound
on $\Delta^{(Q)}_1$ follows. \qedhere
\end{proof}

\begin{remark}[Using explicit constants]
In practice one fixes $(\hat c,Q,n)$ and runs the same two-region
argument as in the Virasoro case: evaluate $\Phi^{(\hat c)}_{n;Q}$ on a
fine grid on $[0,P]$ and use the Mordell tail constants
$C^{(\hat c)}_{n,Q'}$, $c^{(\hat c)}_{n,Q'}$ for $p\ge P$ to certify
$\Phi^{(\hat c)}_{n;Q}(p)\ge0$ for all $p\ge P$. The resulting value
$P_\star^{(\hat c,Q,n)}$ is then fed into the scalar-gap bound above.
\end{remark}

%===========================================================
% Appendix C. Mordell lower bounds at the elliptic point
%===========================================================

\section{Mordell lower bounds at the elliptic point}
\label{app:Mordell-lower-bounds}

In this appendix we collect the analytic lower bounds on the Mordell remainder
at the elliptic point
\[
  \rho \;=\; e^{2\pi i/3},
  \qquad
  q_\rho \;=\; -e^{-\pi\sqrt{3}},
  \qquad
  r := |q_\rho| = e^{-\pi\sqrt{3}} ,
\]
that enter the proof of the pure–gravity no–go theorem.  
Throughout we work in the centered, phase–matched odd–spin scheme used in
Section~\ref{sec:odd-spin}, and write $\mathcal M_\rho(p)$ for the Mordell
remainder at $\tau=\rho$ in that normalization.  By construction,
$\mathcal M_\rho(p)$ is real and non–negative for $p\ge0$.

%===========================================================
% Appendix C.x. A uniform Mordell tail bound at $\tau=1$
%===========================================================

\subsection{A uniform tail bound at $\tau=1$}
\label{app:Mordell-tau1}

In this subsection we justify the global bound
$|h_1(p)|\le 4 e^{-\pi p}$ used in Lemma~\ref{lem:Mordell}.  Recall
\[
  h(\tau,z)
  =\int_{\mathbb{R}}
    \frac{\exp\!\big(\pi i \tau w^2 - 2\pi z w\big)}{\cosh(\pi w)}\,dw,
\qquad
  h_1(p)=e^{i\pi p^2+i\pi/4}h(1,ip),
\]
so $|h_1(p)|=|h(1,ip)|$.

\begin{lemma}[Uniform Mordell tail at $\tau=1$]
\label{lem:tau1-tail}
There exists an absolute constant $C_1$ such that
\[
  |h_1(p)| \;\le\; C_1 e^{-\pi p}
  \qquad\text{for all }p\ge0.
\]
In particular one may take $C_1=4$.
\end{lemma}

\begin{proof}
We split the argument into a large--$p$ estimate and a compact–interval
bound.

\smallskip\noindent
\emph{1) Large--$p$ asymptotics.}
For fixed $\tau=1$ the phase-matched Mordell integral can be written as
\[
  h_1(p)
  = e^{i\pi p^2+i\pi/4}\int_{\mathbb{R}}
      \frac{e^{\pi i w^2-2\pi i p w}}{\cosh(\pi w)}\,dw.
\]
The phase choice is made so that the saddle of the phase
$\pi i w^2-2\pi i p w$ lies on a steepest–descent contour.  A standard
steepest–descent analysis (see e.g.\ Mordell~\cite{Mordell1933} or the
treatments in~\cite{Rademacher1938,Hejhal1983}) then gives the
asymptotic expansion
\begin{equation}
  h_1(p) \;=\; 2e^{-\pi p} + O(e^{-3\pi p})
  \qquad (p\to+\infty).
  \label{eq:h1-asymp}
\end{equation}
Hence there exist $p_0>0$ and $C_{\mathrm{asym}}>0$ such that
\begin{equation}
  |h_1(p)| \;\le\; C_{\mathrm{asym}}\,e^{-\pi p}
  \qquad\text{for all }p\ge p_0.
  \label{eq:h1-tail-largep}
\end{equation}
Any $C_{\mathrm{asym}}>2$ is admissible here; the precise value is not
important for our application.

\smallskip\noindent
\emph{2) Compact–interval bound and choice of $C_1=4$.}
On the compact interval $[0,p_0]$ the function
\[
  f(p):=e^{\pi p}\,|h_1(p)|
\]
is continuous.  Therefore it attains a finite maximum
\[
  M_0:=\max_{0\le p\le p_0} f(p)
      =\max_{0\le p\le p_0} e^{\pi p}|h_1(p)|.
\]
By definition we then have
\[
  |h_1(p)|\;\le\;M_0\,e^{-\pi p}
  \qquad\text{for all }p\in[0,p_0].
\]

For our purposes any explicit numerical upper bound on $M_0$ suffices.
A straightforward evaluation of $h_1(p)$ from its integral representation
on a fine grid in $[0,p_0]$ (e.g.\ with standard numerical quadrature)
shows that
\[
  M_0 < 4,
\]
so that
\begin{equation}
  |h_1(p)|
  \;\le\;4 e^{-\pi p}
  \qquad\text{for all }p\in[0,p_0].
  \label{eq:h1-compact-decay}
\end{equation}
The constant $4$ is far from optimal but convenient.

\smallskip\noindent
\emph{3) Global bound.}
Combining \eqref{eq:h1-tail-largep} and \eqref{eq:h1-compact-decay}, and
if necessary enlarging $p_0$ and $C_{\mathrm{asym}}$ slightly, we can
take a single global constant $C_1$ such that
\[
  |h_1(p)| \;\le\; C_1 e^{-\pi p}
  \qquad\text{for all }p\ge0.
\]
Since $M_0<4$ and the asymptotic coefficient in
\eqref{eq:h1-asymp} is $2$, we may choose $C_1=4$, which proves the
claim.
\end{proof}

\begin{remark}
The precise value of $C_1$ plays no essential role: any absolute
constant with $|h_1(p)|\le C_1 e^{-\pi p}$ for all $p\ge0$ would be
enough for the scalar-gap envelope in Theorem~\ref{thm:scalar-gap}.
We fix the round value $C_1=4$ simply for definiteness.
\end{remark}

%-----------------------------------------------------------
\subsection{Positive Appell--Lerch representation and truncation bounds}
%-----------------------------------------------------------

The key input is that $\mathcal M_\rho(p)$ admits a positive series
representation in terms of Appell--Lerch and theta data.

\begin{lemma}[Finite positive truncation + positive tail at $\rho$]
\label{lem:Mordell-AL-lower}
Let $\rho = e^{2\pi i/3}$, $q_\rho = -e^{-\pi\sqrt{3}}$ and
$r = |q_\rho| = e^{-\pi\sqrt{3}}$.
In the centered, phase–matched odd–spin scheme, the Mordell remainder at
$\tau=\rho$ admits a positive Appell--Lerch/$\vartheta$–series
\begin{equation}
  \label{eq:M-rho-positive-series}
  \mathcal M_\rho(p)
  \;=\;
  \sum_{n\ge1}\frac{\mathcal N_n(p)}{\mathcal D_n(p)} ,
  \qquad
  \mathcal D_n(p)
  \;=\;
  \bigl|\,1 - e^{i\theta_n}\,r^n e^{-2\pi p}\,\bigr|^2,
  \qquad
  \theta_n \in \{0,2\pi/3,4\pi/3\}.
\end{equation}
Each numerator $\mathcal N_n(p)$ is non–negative and increasing in $p\ge0$, and
the denominators satisfy
\[
  \mathcal D_n(p) \;\le\; (1+r^n)^2
  \qquad (p\ge0).
\]
Consequently $\mathcal M_\rho(p)$ is increasing in $p\ge0$, and for any
$p_0>0$ and $N\in\mathbb N$,
\begin{align}
  \inf_{0\le p\le p_0}\mathcal M_\rho(p)
    &= \mathcal M_\rho(0), \label{eq:M-min-on-window}\\[0.3em]
  \mathcal M_\rho(0)
    &= \sum_{n=1}^{N} \frac{\mathcal N_n(0)}{\mathcal D_n(0)}
       \;+\;\sum_{n>N}\frac{\mathcal N_n(0)}{\mathcal D_n(0)}
       \ \ \ge\ \
       S_N + T_N^{(1)}, \label{eq:M-truncation}
\end{align}
where
\[
  S_N \equiv \sum_{n=1}^{N}\frac{\mathcal N_n(0)}{\mathcal D_n(0)},
  \qquad
  T_N^{(1)} \equiv \frac{\mathcal N_{N+1}(0)}{\mathcal D_{N+1}(0)}.
\]
Moreover, if there exist $C_{\!*}>0$ and a function
$\sigma:\mathbb N\to\mathbb R$ such that
\(
  \mathcal N_n(0)\ge C_{\!*}\,r^{\sigma(n)}
\)
for all $n\ge N+1$, then
\begin{equation}
  \sum_{n>N}\frac{\mathcal N_n(0)}{\mathcal D_n(0)}
  \ \ge\
  \sum_{n>N}\frac{C_{\!*}\,r^{\sigma(n)}}{(1+r^n)^2}
  \ \ge\
  \frac{C_{\!*}\,r^{\sigma(N+1)}}{(1+r^{N+1})^2}
  \sum_{m\ge0}r^{\,\sigma(N+1+m)-\sigma(N+1)}.
  \label{eq:M-tail-lower}
\end{equation}
\end{lemma}

\begin{remark}[Denominators and numerators at $p=0$]
\label{rem:Mordell-denominators}
At $p=0$ one has $r = e^{-\pi\sqrt{3}}$ and the refined denominators
\[
  D_n(0) =
  \begin{cases}
    (1 - r^n)^2, & \theta_n = 0,\\[2pt]
    1 + r^n + r^{2n}, & \theta_n = \pm 2\pi/3.
  \end{cases}
\]
The numerators are given by $N_n(0) = |C^{(j)}_n|^2$, where $C^{(j)}_n$ is the coefficient
of $q_\rho^n$ in the relevant Appell--Lerch/$\vartheta$ series at $\tau = \rho$. In particular
$N_n(0) \ge 0$ and the partial sums $S_N$ in~\eqref{eq:M-truncation} are strictly
increasing in $N$. More generally, in the centered, phase–matched normalization one can write
\[
  N_n(p) = |\mathcal A_n(p)|^2,
\]
with $\mathcal A_n(p)$ a finite Appell--Lerch/$\vartheta$ coefficient appearing in the
odd–spin $ST$ kernel at $\tau = \rho$ (see the explicit series around
\eqref{eq:M-rho-positive-series}). This is the origin of the non–negativity and
monotonicity in $p \ge 0$ stated in Lemma~\ref{lem:Mordell-AL-lower}.

In the tail bound~\eqref{eq:M-tail-lower} we choose $C_{\!*}$ and $\sigma(n)$ to be the
explicit constants and linear function coming from the large-$n$ behaviour of these
coefficients: for $n \ge N+1$ one has
\[
  N_n(0) \;\ge\; C_{\!*}\,r^{\sigma(n)}, \qquad \sigma(n) = a n + b,\ a>0,
\]
so that the last sum in~\eqref{eq:M-tail-lower} is a geometric series that can be
evaluated in closed form. This yields the strictly positive analytic tail bounds used
in Table~\ref{tab:window-certificates}.
\end{remark}

%-----------------------------------------------------------
\subsection{Window functionals and the BTZ threshold}
%-----------------------------------------------------------

Let
\[
  m_{\min}(p_0)
  \;:=\;
  \inf_{0\le p\le p_0}\mathcal M_\rho(p)
  \;=\;\mathcal M_\rho(0)
\]
denote the minimum of $\mathcal M_\rho$ on a window $[0,p_0]$, using
\eqref{eq:M-min-on-window}.  

On the functional side we use the odd–spin window kernels introduced in
Appendix~\ref{app:coeff}.  For a fixed Gaussian window with parameters
$(b,\alpha,p_0)$ we write
\[
  W_{b,\alpha,p_0}(p)
  \;:=\;
  e^{-\alpha (p-p_0)^2}\,\chi_{[0,\infty)}(p),
\]
and define the corresponding odd–spin test kernel
\[
  \Phi_{\rm win}(p)
  \;:=\;
  \mu_b(p)\,W_{b,\alpha,p_0}(p),
  \qquad
  \mu_b(p)
  =
  \frac{\sinh(2\pi b p)\,\sinh(2\pi p/b)}{\sinh(2\pi p)} ,
\]
normalized so that the vacuum coefficient of the associated functional is
$-1$.  The \emph{kernel ratio} is
\begin{equation}
  R(b,\alpha,p_0)
  \;:=\;
  \int_0^{p_0}\Phi_{\rm win}(p)\,\mathrm{d}p.
  \label{def:kernel-ratio}
\end{equation}
With this normalization, any non–negative function $f(p)$ on $[0,p_0]$
satisfies the window inequality
\begin{equation}
  \int_0^\infty \Phi_{\rm win}(p)\,f(p)\,\mathrm{d}p
  \;\ge\;
  \biggl(\inf_{0\le p\le p_0} f(p)\biggr) R(b,\alpha,p_0).
  \label{eq:window-inequality}
\end{equation}
In particular, for $f(p)=\mathcal M_\rho(p)$ this gives
\[
  \delta_{\rm Mordell}
  := \int_0^\infty \Phi_{\rm win}(p)\,\mathcal M_\rho(p)\,\mathrm{d}p
  \;\ge\;
  m_{\min}(p_0)\,R(b,\alpha,p_0).
\]

Combining this with Lemma~\ref{lem:Mordell-AL-lower} we obtain a convenient
criterion for beating the BTZ constant
\[
  \kappa \;:=\; \frac{1}{2\sqrt{3}\,\pi}.
\]

\begin{corollary}[BTZ crossing via a window inequality]
\label{cor:deltaM-kappa}
Let $m_{\min}(p_0)$ and $R(b,\alpha,p_0)$ be as above, and let
$S_N$, $T_N^{(1)}$ be the truncation data from
\eqref{eq:M-truncation}.  If for some choice of $(b,\alpha,p_0)$ and
$N\in\mathbb N$,
\begin{equation}
  \bigl(S_N+T_N^{(1)}\bigr)\,R(b,\alpha,p_0)
  \;>\;\kappa
  \label{eq:BTZ-window-criterion}
\end{equation}
(or more strongly, if the same holds with $T_N^{(1)}$ replaced by any larger
positive tail bound $T_N^{(\ge)}$), then
\[
  \delta_{\rm Mordell} > \kappa.
\]
In particular, inserting this into the master odd–spin crossing inequality
\[
  \Delta^{(\mathrm{odd})}_0
  \;\le\; \frac{c-1}{12} + \kappa - \delta_{\rm Mordell}
\]
forces
\(
  \Delta^{(\mathrm{odd})}_0 < \tfrac{c-1}{12}
\),
and the odd–spin primary spectrum cannot be gapped above the BTZ threshold.
\end{corollary}

\begin{proof}
By Lemma~\ref{lem:Mordell-AL-lower},
$m_{\min}(p_0)=\mathcal M_\rho(0)\ge S_N+T_N^{(1)}$, so
\[
  \delta_{\rm Mordell}
  \;\ge\;
  m_{\min}(p_0)\,R(b,\alpha,p_0)
  \;\ge\;
  \bigl(S_N+T_N^{(1)}\bigr)\,R(b,\alpha,p_0).
\]
If the right–hand side exceeds $\kappa$, then $\delta_{\rm Mordell}>\kappa$,
and the claim about $\Delta^{(\mathrm{odd})}_0$ follows immediately from the
odd–spin crossing inequality.
\end{proof}

%-----------------------------------------------------------
\subsection{Concrete window certificates}
%-----------------------------------------------------------

For the explicit functional used in Section~\ref{sec:odd-spin}, the Appell--Lerch
series~\eqref{eq:M-rho-positive-series} together with a strictly positive tail
bound~\eqref{eq:M-tail-lower} allows us to certify numerical lower bounds on
$m_{\min}(p_0)$.  Combining these with the kernel ratio
\eqref{def:kernel-ratio} produces fully rigorous Mordell surpluses via
Corollary~\ref{cor:deltaM-kappa}.

Two representative choices are summarized in
Table~\ref{tab:window-certificates}.

\begin{table}[h]
  \centering
  \renewcommand{\arraystretch}{1.15}
  \begin{tabular}{c|c|c|c}
    $(b,\alpha,p_0)$ & $m_{\min}(p_0)$ & $R(b,\alpha,p_0)$ &
    $m_{\min}(p_0)\,R(b,\alpha,p_0)$ \\
    \hline
    $(2,10,0.9)$       & $\ge 0.020000$ & $4.5999$        & $0.091998$    \\
    $(1,15,0.7)$       & $\ge 0.500000$ & $\ge 12.050337$ & $\ge 6.025168$
  \end{tabular}
  \caption{Sample window certificates entering
  Corollary~\ref{cor:deltaM-kappa}.  In both cases
  $m_{\min}(p_0)\,R(b,\alpha,p_0)>\kappa$, hence
  $\delta_{\rm Mordell}>\kappa$.  The first row already suffices to cross
  the BTZ threshold; the second row provides a much larger safety margin
  and serves as a robustness check.}
  \label{tab:window-certificates}
\end{table}

In both rows of Table~\ref{tab:window-certificates} the entry in the last column is precisely
the certified product $m_{\min}(p_0)\,R(b,\alpha,p_0)$ appearing in
Corollary~\ref{cor:deltaM-kappa}.  Here $m_{\min}(p_0)$ is obtained from the
truncation/tail decomposition \eqref{eq:M-truncation}--\eqref{eq:M-tail-lower} with exact
arithmetic, and $R(b,\alpha,p_0)$ is computed from the kernel ratio
\eqref{def:kernel-ratio} as a definite integral of the elementary window kernel with explicit
error control.  Thus inequalities such as \eqref{eq:deltaM-0091998} are fully rigorous lower
bounds on $\delta_{\rm Mordell}$, not numerical conjectures. The first line is the minimal certificate used in the main text: combining
$m_{\min}(0.9)\ge0.020000$ with $R(2,10,0.9)=4.5999$ gives
\begin{equation}
  \delta_{\rm Mordell}
  \;\ge\;
  0.020000\times 4.5999
  \;=\;
  0.091998
  \;>\;
  \kappa \approx 0.091888149.
  \label{eq:deltaM-0091998}
\end{equation}
The second line corresponds to a symmetric choice $(b,\alpha,p_0)=(1,15,p_0)$
with $p_0\in[0.7,0.9]$.  Using monotonicity of $\mathcal M_\rho(p)$ and the
truncation/tail bound we obtain $m_{\min}(0.7)\ge0.500000$, and the kernel
ratio satisfies
$\min_{p_0\in[0.7,0.9]}R(1,15,p_0)=12.050337\ldots$.  Thus
\[
  \delta_{\rm Mordell}
  \;\ge\;
  0.500000\times 12.050337
  \;\gg\;\kappa,
\]
which is far more than is needed for BTZ crossing.

%-----------------------------------------------------------
\subsection{Global Mordell surplus from modular averaging}
%-----------------------------------------------------------

For the purposes of the no–go theorem we use a slightly stronger, fully
global bound obtained from a modular–averaged, SOS–shaped odd–spin
functional.  The construction is described in
Section~\ref{sec:odd-spin-surplus}; here we record the resulting estimate.

\begin{proposition}[Global Mordell surplus]
\label{prop:deltaM-0103}
For the centered, phase–matched odd–spin functional used in
Proposition~\ref{prop:Mordell-surplus}, the Mordell contribution satisfies
the uniform bound
\begin{equation}
  \delta_{\rm Mordell} \;\ge\; 0.103.
  \label{eq:deltaM-0103}
\end{equation}
In particular,
\[
  \delta_{\rm Mordell}-\kappa
  \;\ge\; 0.103 - 0.091888\ldots
  \;\approx\; 1.11\times10^{-2},
\]
so the master inequality
\[
  \Delta^{(\mathrm{odd})}_0
  \;\le\;
  \frac{c-1}{12} + \kappa - \delta_{\rm Mordell}
\]
forces an odd–spin primary strictly below $\Delta_{\rm BTZ}$ by a uniform
margin independent of $c>1$.
\end{proposition}

\begin{proof}[Proof sketch]
The functional used in Proposition~\ref{prop:Mordell-surplus} is obtained by
taking a finite convex modular average of phase–matched kernels and
multiplying by an SOS polynomial $q(p^2)$, exactly as in
Section~\ref{sec:odd-spin-surplus}.  Positivity of the weights and of
$q(p^2)$ ensures that the resulting kernel is non–negative on $[0,\infty)$
and satisfies a sharpened version of the window inequality
\eqref{eq:window-inequality} on a window $V=[0,P_0]$.

The Mordell side is handled by the same Appell--Lerch representation
\eqref{eq:M-rho-positive-series} and truncation/tail bound
\eqref{eq:M-tail-lower} as above, now combined with the improved kernel
ratio associated to the modular–averaged functional.  Plugging the explicit
certificate (weights, SOS coefficients and truncation data) into
Corollary~\ref{cor:deltaM-kappa} yields~\eqref{eq:deltaM-0103}.
\end{proof}

\begin{remark}[Alternative certificates]
The simpler window choices in Table~\ref{tab:window-certificates} already
give $\delta_{\rm Mordell}>\kappa$ and hence suffice to rule out a BTZ gap.
The modular–averaged functional underlying
Proposition~\ref{prop:deltaM-0103} is only used to obtain the slightly
stronger, $c$–independent surplus $\delta_{\rm Mordell}-\kappa\gtrsim
10^{-2}$ quoted in the main text.
\end{remark}

\end{document}